\documentclass{m2an_blank}
\usepackage{graphicx}
\usepackage{amsmath,amssymb, amsthm,amsbsy}
\usepackage{url}
\usepackage{pstricks}
\usepackage{bbm}
\usepackage{mathtools}
\mathtoolsset{showonlyrefs}

\newcommand{\one}{\mathbbm{1}}

\newcommand{\D}{\mathcal D}
\newcommand{\EE}{\mathbb E}

\newcommand{\Gb}{{\mathbf G}}
\newcommand{\LL}{\mathcal L}
\newcommand{\NN}{\mathbb N}
\newcommand{\PP}{\mathbb P}

\newcommand{\qb}{{\mathbf q}}
\newcommand{\qtb}{\widetilde{\qb}}
\newcommand{\Qb}{{\mathbf Q}}

\newcommand{\Qtb}{\widetilde{\Qb}}

\newcommand{\rh}{\hat{r}}
\newcommand{\RR}{\mathbb R}
\newcommand{\SSph}{{\mathbb S}}
\newcommand{\ub}{\mathbf{u}}
\newcommand{\Vt}{\widetilde{V}}
\newcommand{\Vtb}{\mathbf{{\Vt}}}
\newcommand{\Vh}{\widehat{{V}}}
\newcommand{\Vavg}{\mathring{V}}
\newcommand{\Vavgb}{\mathring{\mathbf{V}}}
\newcommand{\Vhm}{{\Vh_{\mathrm{min}}}}
\newcommand{\Vb}{{\mathbf V}}
\newcommand{\vmax}{{v_{\mathrm{max}}}}

\newcommand{\xm}{{x_{\rm min}}}
\newcommand{\Wb}{{\mathbf W}}

\newcommand{\ZZ}{{\mathbb Z}}

\newcommand{\eb}{\mathbf{e}}

\newcommand{\is }{I_{\rm slow}}
\newcommand{\ifa}{I_{\rm fast}}
\newcommand{\iex}{I_{\rm rem}}
\newcommand{\ime}{I_{\rm mech}}
\newcommand{\ima}{I_{\rm Mark}}
\newcommand{\iext}{I_{\rm add}}
\newcommand{\tint}{{\{0 \leq t \leq T\}}}
\newcommand{\vb}{\mathbf{v}}
\newcommand{\vce}{{\overline{v}}}

\newcommand{\vcen}{{\overline{\vb}}}
\newcommand{\Vcen}{{\overline{\Vb}}}
\newcommand{\pconv}{{\stackrel{p}{\longrightarrow}}}
\newcommand{\lconv}{{\stackrel{\LL}{\longrightarrow}}}

\newcommand{\eps}{\varepsilon}

\DeclareMathOperator{\tr}{tr}
\DeclareMathOperator{\divop}{div}

\DeclareMathOperator*{\esssup}{ess\,sup}

\numberwithin{equation}{section}

\theoremstyle{plain}
\newtheorem{theorem}{Theorem}[section]
\newtheorem{lemma}{Lemma}[section]

\newtheorem{definition}{Definition}[section]

\theoremstyle{remark}
\newtheorem{remark}{Remark}[section]

\begin{document}
\title{Derivation
of Langevin Dynamics in a Nonzero Background Flow Field}
\date{\today}
\author{Matthew Dobson}\sameaddress{3, 4}
\secondaddress{now at: Department of Mathematics and Statistics,
710 N.~Pleasant Street,
University of Massachusetts,
Amherst, MA 01003-9305, USA
\email{dobson@math.umass.edu}
}
\author{Fr\'ed\'eric Legoll}\address{Laboratoire Navier - Ecole des Ponts ParisTech,
6 et 8 avenue Blaise Pascal,
Cit\'e Descartes - Champs sur Marne,
77455 Marne la Vall\'ee Cedex 2, France
\email{legoll@lami.enpc.fr}
}
\secondaddress{INRIA Rocquencourt, MICMAC team-project,
{\footnotesize Domaine de Voluceau, B.P. 105,
  78153 Le Chesnay Cedex, France}
}
\author{Tony Leli\`evre}
\sameaddress{3}
\secondaddress{CERMICS - Ecole des Ponts ParisTech,
6 et 8 avenue Blaise Pascal,
Cit\'e Descartes - Champs sur Marne,
77455 Marne la Vall\'ee Cedex 2, France
\email{\{lelievre,stoltz\}@cermics.enpc.fr}
}
\author{Gabriel Stoltz}\sameaddress{3, 4}

\begin{abstract}
We propose a derivation of a nonequilibrium Langevin dynamics for a
large particle immersed in a background flow field.  A single large
particle is placed in an ideal gas heat bath composed of point
particles that are distributed consistently with the background flow
field and that interact with the large particle through elastic
collisions.  In the limit of small bath atom mass, the large particle
dynamics converges in law to a stochastic dynamics.  This
derivation follows the ideas of~\cite{durr83,cald89,durr80} and
provides extensions to handle the nonzero background flow.  The
derived nonequilibrium Langevin dynamics is similar to the dynamics
in~\cite{mcph01}.  Some numerical experiments illustrate the use of
the obtained dynamic to simulate homogeneous liquid materials under
shear flow.
\end{abstract}

\maketitle

\section{Introduction}

Molecular dynamics simulations have been increasingly used to bring
atomistic accuracy to macroscopic fluid
models~\cite{ocon95,werd05,hadj05,ren_05}.  One example is the
computation of the constitutive relation between the strain rate
$A=\nabla \ub$ and the stress tensor $\boldsymbol{\sigma}(\nabla \ub, T)$ in
complex fluids at temperature $T$, where one uses a microscopic
simulation to determine the closure relation for the continuum
equations~\cite{lebr_visco}.   We thus wish to simulate
molecular systems at temperature $T$ that are subject to a steady,
nonzero macroscopic flow, and the goal of this paper is the
derivation of a dynamics to sample such states. 

Simulation of molecular systems with a nonzero background flow is one
goal of nonequilibrium molecular dynamics (NEMD)
techniques~\cite{cicc05,evan07}.  Several strategies have been
proposed to sample molecular system under nonzero flow: for example,
the SLLOD and g-SLLOD equations of
motion~\cite{evan07,edwa06,todd07,tuck97} or dissipative particle
dynamics~\cite{sodd03}.  Typically, these are used in conjunction with
consistent boundary conditions such as the
Lees-Edwards boundary conditions~\cite{evan07} in the case of shear
flow or the Kraynik-Reinelt boundary conditions~\cite{todd99} for
elongational flow.  The basic equations of motion for these algorithms
typically exhibit energy growth and need some form of modification if
one wants temperature control for the system.  Two common choices are
the isokinetic Gaussian thermostat and the Nos\'e-Hoover thermostat.
It has been shown that the Nos\'e-Hoover dynamics is non-ergodic for
the NVT ensemble~\cite{lego07,lego09}, and furthermore we are not
aware of an analysis of the ergodicity of the Gaussian thermostat for
these types of nonequilibrium molecular systems.  Instead, we 
work in a framework that leads to stochastic equations that have a 
Langevin-type thermostat.
In the case of no background flow, the standard Langevin dynamics has
long been used to sample the NVT ensemble, and it can be shown to be
ergodic.  

In the following, we \emph{derive} a system of equations for a single
particle in a microscopic heat bath with a nonzero, constant-in-time
background flow in such a way that velocity gradient control and
temperature control are incorporated at the same point in the model.
The resulting system of equations, which we refer to as the 
\emph{nonequilibrium Langevin dynamics}, is
\begin{equation}
\begin{split}
\label{ou_flow_intro}
d \Qb &= \Vb dt, \\
M d \Vb &= - \gamma (\Vb-A \Qb) dt + \sigma d\Wb,
\end{split}
\end{equation}
where $(\Qb,\Vb)$ are the particle's position and velocity, $M$ is the
mass of the particle, $\Wb$ is a standard Brownian motion, $A$
is the homogeneous strain rate, and $\gamma$ and $\sigma$ are scalar
constants satisfying the fluctuation-dissipation relation
\begin{equation}
\label{fdr}
\gamma = \frac{1}{2} \sigma^2 \beta
\end{equation}  
where $\beta = \frac{1}{k_B T}$ is the inverse temperature.  This system
of equations is a generalization of Langevin dynamics, which we recover
in the limit $A \rightarrow 0.$  In the
sequel, we follow the
derivation of D\"urr, Goldstein, and Lebowitz
(DGL)~\cite{durr80,durr83} who consider the case of a heat bath that
has zero background flow.  

We briefly summarize the ingredients of the mechanical model used 
in the derivation.  A full description is given in Section 2.
The microscopic mechanical model consists
of a single large particle
immersed in a bath composed of infinitely many small bath atoms.  
The mass of the large particle, $M,$ is held constant while
the mass of an individual bath atom, $m,$ is a parameter of the
mechanical model.  The large particle moves in a ballistic trajectory
until it collides with a bath atom, at which point an elastic collision
occurs according to~\eqref{collision2} below, causing a jump in the velocity.  
The heat bath is constructed in such a way that most of these jumps are independent
events distributed according to a velocity measure centered around the
background flow.   
More precisely, the bath atoms have a random initial velocity 
distribution that is centered around the desired background flow~\eqref{bath_meas2}-\eqref{fvq},
with its mean relative kinetic energy proportional to the macroscopic
temperature.    A microscopic dynamics for the bath
atoms is chosen so that the velocity distribution of the bath
atoms is preserved (one choice is given by~\eqref{nonHam1}), up until they collide with the particle.  A typical bath atom has velocity much
larger than the velocity of the large particle, and such a bath atom
will collide at most once with the large particle (a fact made rigorous in
Appendix~\ref{norecollide}).
The nonequilibrium Langevin dynamics is then derived as the limiting dynamics
of the large particle when
$m \rightarrow 0$ (Theorem~\ref{thm:main}).  
  
The full description of the original DGL model
as well as two approaches to incorporating background flow are given
in Section~\ref{sec:model} culminating in the main convergence result,
Theorem~\ref{thm:main}.  The proof of convergence of the heat bath
model to the derived stochastic equations~\eqref{ou_flow_intro} is
carried out  in the Appendices.  The proof is structured as
in~\cite{durr80}, with added arguments to control how the flow in the
heat bath affects the error growth.  In Section~\ref{sec:num} we
include numerical results showing the application of the derived
equations to computing the shear stress of a Lennard-Jones fluid.

\section{Models for non-uniform background flow}
\label{sec:model}
In the following, we consider a system in $\RR^d$ (for $d=2 \text{
or } 3$)  composed of a single, distinguished large {particle}
immersed in a heat bath of light {atoms} that have mean velocity
$A\qb$ at point $\qb.$  We derive the equations of motion of the large
particle in the limit as the mass of the individual bath atoms goes to
zero.  We note that we use the terms `particle' and `atoms' to
differentiate the two types of objects, but the heat bath could well
be composed of light molecules.  We follow the arguments of D\"urr,
Goldstein, and Lebowitz (DGL)~\cite{durr80}, who consider a single
large particle placed in an infinite, constant-temperature heat bath
with zero background flow (note that the case of a constant uniform
background velocity is equivalent to the case of zero background flow
by a change of coordinates).  The only forces which act on the large
particle are due to the heat bath (see Section~\ref{sec:gen} for
extensions to multiple particles and more general interactions).  In
the limit of small bath atom mass for zero background flow, DGL
recover the Langevin dynamics 
\begin{equation}
\label{lang1}
\begin{split} 
d \Qb &= \Vb dt, \\
M d \Vb &= - \gamma \Vb dt + \sigma d\Wb,
\end{split} 
\end{equation}
where $\Wb$ denotes a $d$-dimensional Brownian motion, $\Qb,\Vb \in
\RR^d$ denote the position and velocity of the large particle, and
$\gamma, \sigma \in \RR$ are scalar constants that are determined by
the parameters of the large particle and the heat bath and that
satisfy the fluctuation-dissipation relation~\eqref{fdr}.  Both the
drift and diffusion terms in the velocity equation are caused by the
elastic collisions of the large particle with the bath atoms, and the
randomness stems only from the random initial configuration of the
bath.  In this paper we extend the work of DGL~\cite{durr80,durr83} to
the case of a nonzero background flow.

In Section~\ref{mechdef}, we review the construction of the heat bath
and trajectory of the large particle in the zero-flow case
of~\cite{durr80} before outlining in the remaining subsections 
two possible approaches for
extending this to the case of nonzero background flow.  The first approach, 
outlined in Section~\ref{sec:lam1}, applies to shear flows modeled by
a heat bath with one or multiple unidirectional laminar flows. 
In this approach, the limiting nonequilibrium dynamics is not
the dynamics given in~\eqref{ou_flow_intro}.  For reasons outlined
in Section~\ref{sec:lam1}, this first approach is not the one we focus on in this paper.
We then describe a second approach in Section~\ref{meas_dyn} involving
a modified non-Hamiltonian dynamics for the bath atoms, which yields in the limit
$m \rightarrow 0$ the nonequilibrium
Langevin dynamics~\eqref{ou_flow_intro} for general incompressible background flows.  
Note that, in contrast to the first approach, this second approach is not
restricted to shear flows.

\medskip

In the following, we use bold notation for vectors and normal weight
for scalars and matrices.  Capital letters are used to distinguish the
large particle's position and velocity $(\Qb, \Vb)$ from that of a
generic bath atom $(\qb, \vb)$. 
 
\subsection{Heat bath with zero background flow}
\label{mechdef}

Here we recall the model and results of~\cite{durr80}.  The bath
consists of infinitely many light atoms each with position $\qb \in
\RR^d$ and velocity $\vb \in \RR^d$.  All bath atoms have the same
mass, $m>0,$ and zero radius.  The initial bath configuration is drawn
from a Poisson field (whose definition we recall below) with an
$m$-dependent measure given by
\begin{equation}
\label{bath_meas1}
d \mu_m(\vb,\qb) = \lambda_m f_m(\vb) \, d \qb \, d \vb,
                           \qquad \qb,\vb \in \RR^d,
\end{equation}
where 
\begin{equation*}
\lambda_m = m^{-1/2} \lambda
\end{equation*} 
is the expected number of atoms per unit volume and $f_m$ is the
scaled probability distribution on the velocities.  The velocity
probability distribution scales like
\begin{equation*}
f_m(\vb) = m^{d/2} f(m^{1/2}\vb),
\end{equation*} 
which means that for a single bath atom, the expected speed
$\EE_m(|\vb|) = \int_{\RR^d} |\vb| f_m(\vb) \, d \vb$ is proportional to
$m^{-1/2}.$  This scaling ensures that the average kinetic energy per
atom is constant in the limit $m \rightarrow 0.$
\begin{figure}
\centerline{\input{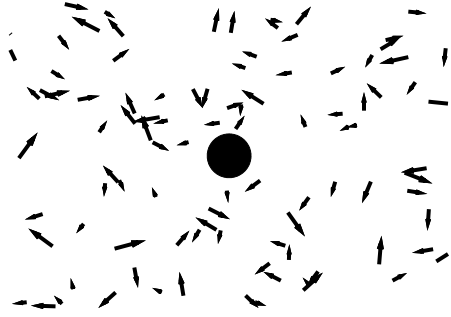}}
\caption{\label{fig:uniformbath}
A large particle (of radius $R=2$), surrounded by a heat bath with zero background flow,
whose atoms are uniformly distributed in space.  The velocity of an
atom is distributed according to~\protect\eqref{gauss1}, with
$\beta=1,$ $m=1,$ and $\lambda=\frac{1}{16}.$ }
\end{figure}
The distribution $f(\vb)$ is assumed to be rotationally invariant.  
The quantities
\begin{equation}
\label{phi_i_2}
\Phi_i = \frac{1}{2} \int_{\RR^d} |v_1|^i f(\vb) \, d \vb,
	\qquad \qquad \text{for } i = 1, \dots, 4,
\end{equation} 
where $v_1$ denotes the first component of $\vb,$ are assumed to be
finite.  For a set $S \subset \RR^d \times \RR^d$ in phase space, with
measure denoted by $\mu_m(S),$ the number of atoms in $S$, denoted
$N_m(S),$ is a Poisson random variable with parameter $\mu_m(S)$, so
that for $k\geq 0,$ 
\begin{equation*}
\PP( N_m(S) = k) = \frac{\mu_m(S)^k}{k!} e^{-\mu_m(S)}.
\end{equation*}
Figure~\ref{fig:uniformbath} shows one possible realization of the
heat bath with the choice
\begin{equation}
\label{gauss1}
f(\vb) = Z^{-1} \exp\left(- \frac{\beta}{2} |\vb|^2\right),
\end{equation}
where $Z = \left(\frac{2 \pi}{\beta}\right)^{d/2}$ is the
normalization constant.

As soon as the initial condition of the heat bath process has been chosen,
the evolution of the system is deterministic.
The bath is initialized and a large particle with a finite radius $R$
and mass $M$ is placed in the bath at $(\Qb(0),\Vb(0)),$ where we note
that the initial condition is independent of $m.$  
Any bath atoms that are underneath
the large particle at the initial time, that is, those such that
$|\qb(0)-\Qb(0)| \leq R$, are removed from the bath (we recall
that the bath atoms have zero radius).  
The bath atoms have no self-interaction and, aside from collisions
with the large particle, obey the dynamics
\begin{equation} 
\label{ballistic}
\begin{split}
d \qb &= \vb dt, \\
d \vb &= 0.
\end{split} 
\end{equation} 
Likewise, the large particle's position and velocity
evolve according to
\begin{equation}
\label{particle_dynamics} 
\begin{split} 
d\Qb_m &= \Vb_m \, dt, \\
d\Vb_m &= 0.
\end{split}
\end{equation} 
We explicitly denote the $m$-dependence of the mechanical process, 
while for notational convenience, we
omit the $m$-dependence of the bath coordinates.

When a bath atom encounters the large particle, it undergoes an
elastic collision.  
For a single collision, let $\eb_{\rm n}$ denote the unit vector in
the direction from the bath atom to the particle center.  Let
$\Vb_{\rm n} = (\Vb \cdot \eb_{\rm n}) \eb_{\rm n}$ and $\Vb_{\rm t} =
\Vb - \Vb_{\rm n}$ denote the velocity of the large particle in the
normal and tangential directions.  Similarly, define $\vb_{\rm n} =
(\vb \cdot \eb_{\rm n}) \eb_{\rm n}$ and $\vb_{\rm t} = \vb - \vb_{\rm
n}$ for the bath atoms.  We distinguish between the component of
$\Vb$ in the normal direction $\Vb_{\rm n} = (\Vb \cdot \eb_{\rm n})
\eb_{\rm n}$ and its magnitude $V_{\rm n} = \Vb \cdot \eb_{\rm n}.$
The elastic collision rule is
\begin{equation}
\label{collision2}
\begin{split}
\Vb'_{\rm t} &= \Vb_{\rm t}, \qquad \qquad \vb'_{\rm t} = \vb_{\rm t}, \\
\Vb'_{\rm n} &= \hphantom{-}\frac{M-m}{M+m} \Vb_{\rm n} + \frac{2 m}{M+m} \vb_{\rm n}, \\ 
\vb'_{\rm n} &= -\frac{M-m}{M+m} \vb_{\rm n} + \frac{2 M}{M+m} \Vb_{\rm n}, \\ 
\end{split}
\end{equation}
where primes denote the after-collision velocities.  The generated
trajectory $(\Qb_m,\Vb_m)$ is called the \emph{mechanical process}.
It is possible to only consider the case of single collisions since it
has been shown that pathologies such as multiple simultaneous
collisions or an infinite number of collisions in a finite time
interval have zero probability~\cite[Appendix]{durr80} (see also
Appendix~\ref{wellposed}).  As a result
of the above, the mechanical process is well-defined on the finite
interval~$[0,T].$

Each initial condition of the heat bath corresponds to a realization
of the mechanical process $(\Qb_m, \Vb_m).$
The velocity $\Vb_m$ is defined to have right-continuous jumps so that the 
trajectory is a c\`adl\`ag function, that is,
a piecewise continuous function that
is everywhere right-continuous and has a left-hand limit.
For a fixed initial condition of the heat bath, the particle
trajectory $(\Qb_m,\Vb_m)$ is a deterministic
process; however, since the bath's initial condition is a random
variable, the particle's position and velocity are also random
variables.

\begin{remark}
\label{rem:notMarkov}
The trajectory $(\Qb_m, \Vb_m)$ is not a Markov
process.  It would only be Markov if the rate of collisions with bath
atoms did not depend on the history of the particle; however, this is
not the case because of two effects.  First, recollisions are possible
with bath atoms that are moving sufficiently slowly.  Second, certain
collisions, called `virtual collisions', with slow-moving atoms are
made impossible due to the particle `sweeping out' a path behind it.
The large particle creates a wake behind itself, and whenever it
decelerates, there is a minimum delay before slow moving atoms can hit
it from behind.  To show convergence of the mechanical process to the
limiting Langevin dynamics~\eqref{lang1}, one must show that
these history effects are negligible in the limit $m \rightarrow 0.$
\end{remark}

We now define the appropriate topology in order to precisely state the
convergence result of~\cite{durr80}.  We fix a finite time $T$ and let
$\D([0,T])$ denote the space of c\`adl\`ag functions on the interval $[0,T]$.
For functions $f, g \in \D([0,T]),$ we define the Skorokhod metric
(see for example~\cite{bill99})
\begin{equation}
\label{skor}
\sigma_{\rm {sk}}(f,g) = \inf_{\lambda \in \Lambda} \max \{ 
   \| \lambda - \mathrm{Id} \|_{L^{\infty}([0,T])}, 
   \|f - g \circ \lambda \|_{L^{\infty}([0,T];\RR^d)} \}, 
\end{equation} 
where $\mathrm{Id} : [0,T] \rightarrow [0,T]$ denotes the identity map and
$\Lambda$ is the set of all strictly increasing, continuous bijections
of $[0,T]$.  For vector-valued $f$, we apply the Euclidean
norm in space so that the $L^{\infty}$-norm above is given by 
\begin{equation*}
   \|f \|_{L^{\infty}([0,T];\RR^d)} = \esssup_{t \in [0,T]} | f (t) |.
\end{equation*}
The mechanical process $(\Qb_m, \Vb_m)$ defined above
is in $\D([0,T]).$ 

We now recall the following definitions for the convergence of random 
variables.
\begin{definition}
\label{def:conv}
Let $Z_m$ and $Z$ be random variables in the metric space
$(\D([0,T]),\sigma_{\rm {sk}}).$  We denote convergence in law by
$(Z_m)_\tint \lconv Z_\tint$, that is,
$\EE f (Z_m) \rightarrow \EE f(Z)$ as $m\rightarrow 0$ for all
bounded, continuous functions $f: \D([0,T]) \rightarrow \RR.$
We denote convergence in probability by $(Z_m)_\tint  \pconv (Z)_\tint $,
that is, for all $\eps > 0,$ 
\begin{equation*}
\lim_{m \rightarrow 0} \PP(\{\sigma_{\rm {sk}}(Z_m,Z) > \eps\}) = 0.
\end{equation*} 
\end{definition} 

The following theorem gives not only convergence in law of the
mechanical process to the Langevin dynamics~\eqref{lang1}, but also an
explicit expression for the coefficients of the limiting dynamics:
\begin{equation}
\label{lang_coeff}
\gamma = \frac{4 \lambda R^{d-1} S_{d-1}}{d} \Phi_1, 
  \qquad 
\sigma = \left[\frac{4 \lambda R^{d-1} S_{d-1}}{d} \Phi_3
\right]^{\frac{1}{2}},
\end{equation}
where $S_{d-1}$ is the surface area of the $(d-1)$-sphere
$\SSph^{d-1}$,  $\Phi_i$ is defined in~\eqref{phi_i_2}, $R$ and $M$
are the radius and mass of the large particle, and $\lambda$
is a density parameter of the bath in~\eqref{bath_meas1}.

\begin{theorem}[{\cite[Theorem 2.1]{durr80}}]
\label{thm:dlg}
Consider the mechanical process $(\Qb_m, \Vb_m),$ defined in~\eqref{ballistic},~\eqref{particle_dynamics}, and~\eqref{collision2} with 
bath measure defined in~\eqref{bath_meas1}, and 
let $(\Qb, \Vb)$ denote the solution to the 
Langevin equation~\eqref{lang1} 
with coefficients~\eqref{lang_coeff}.
For any $T>0,$ 
\begin{equation*}
(\Qb_m, \Vb_m)_\tint \lconv (\Qb, \Vb)_\tint
\end{equation*}
as $m \rightarrow 0,$ where convergence is with respect to 
the Skorokhod metric~\eqref{skor} on $\D([0,T])$.
\end{theorem}
The resulting dynamics~\eqref{lang1} satisfies
the fluctuation-dissipation relation~\eqref{fdr}
provided that 
\begin{equation*}
\frac{\Phi_1}{\Phi_3} = \frac{\beta}{2},
\end{equation*}
which holds true for $f(\vb)$ as in~\eqref{gauss1}, for example.
In this case, the temperature of the large particle is identical to
that of the bath.  The main result of this paper,
Theorem~\ref{thm:main}, is the generalization of Theorem~\ref{thm:dlg}
to the case where the heat bath has a nonzero background flow.

\subsection{Generalization to a nonzero background flow}

We generalize the model described 
above by considering a bath measure $d \mu_m(\qb, \vb)$  
where, in contrast to the previous section, the initial distribution 
of the velocities depends on the position.  We give two possible approaches
to the generalization.  In the first, which applies only to shear flow,
the velocities are restricted to laminar profiles.  For the second,
which applies to general incompressible flows, 
we change the microscopic dynamics that the bath atoms satisfy in
order to preserve the distribution of atoms.
In both cases, we require that the bath measure and 
dynamics satisfy the following hypotheses:
\begin{gather}
\label{marge_vb}
\tag{H1}
\textrm{The average velocity at 
point $\qb^*$ is equal to $A \qb^*$ in the sense that } \\
\textrm{for all } \qb^* \in \RR^d, \,  
\lim_{\eps \to 0} \frac{\displaystyle \int_{\RR^d} \int_{{\mathcal B}(\qb^*,\eps)} \vb \, \mu_m(d \qb, d \vb) } {\displaystyle \int_{\RR^d} \int_{{\mathcal B}(\qb^*,\eps)}   \, \mu_m(d \qb, d \vb) } 
\textrm{ exists and is equal to }
A \qb^*, \\
\label{invariance}
\tag{H2}
\text{The measure } d \mu_m(\qb,\vb) 
\text{ is invariant under the bath dynamics in the absence of
collisions,}
\end{gather}
where $\mathcal{B}(\qb, \eps)$ denotes the ball of radius $\eps$ about 
$\qb.$
The invariance of the bath measure means that, except for the effects
described in Remark~\ref{rem:notMarkov}, the large particle
experiences collisions with a time-independent rate.  We note that 
satisfying~\eqref{marge_vb} and~\eqref{invariance} is not automatic.  For example, if
the velocities $\vb$ have a Gaussian distribution centered around
$A \qb$ and the atoms follow ballistic trajectories~\eqref{ballistic},
then the system satisfies~\eqref{marge_vb} but not~\eqref{invariance}.
In both generalizations described below, the bath atoms interact with the large
particle only through elastic collisions as in~\eqref{collision2}.

In Section~\ref{sec:lam1}, we describe a laminar flow model for the
specific case where the background flow is a shear.  All bath atoms are
restricted to move in one of the coordinate directions, following
otherwise ballistic trajectories.  The resulting dynamics is not the
same as the nonequilibrium Langevin dynamics~\eqref{ou_flow_intro}
that form the main focus of this work.  
In Section~\ref{meas_dyn} we do not restrict
the velocity directions, and we choose a modified (non-Hamiltonian)
dynamics for the bath atoms which leaves the chosen bath measure
invariant for any traceless $A.$    This
approach leads to the nonequilibrium Langevin dynamics~\eqref{ou_flow_intro} 
as limiting equations for the
large particle in the limit $m \rightarrow 0.$

\subsection{Laminar flow models}
\label{sec:lam1}

For this subsection, 
we restrict ourselves to a heat bath in $\RR^3$ under
shear flow with the specific strain rate
\begin{equation}
\label{shear_mat}
A = \left[ \begin{array}{ccc}
0&s&0\\
0&0&0\\
0&0&0\\
\end{array} \right],
\end{equation}
for some given $s \in \RR.$
We describe two variants of a laminar flow model, which both create
a background shear flow by restricting the atom velocities
to be parallel to the coordinate axes.

\subsubsection{Single laminar flow}
We enforce a steady shear flow for the bath atoms by choosing initial
velocities that are nonzero only in the first coordinate, $\eb_1 =
[1,0,0]^T.$  We choose the initial configuration of the bath atoms according
to a Poisson field with measure
\begin{equation}
\label{fvq_shear}
d \mu_m(\qb,\vb) = \lambda_m Z^{-1} m^{1/2} \exp\left( -\frac{\beta}{2} 
       m(v_1 - s q_2)^2 \right) \delta(v_2) \delta(v_3) \, d \qb \, d \vb ,
\end{equation}
where $Z = \left(\frac{2 \pi}{\beta}\right)^{1/2}$ and $\delta$
denotes the Dirac distribution (see Figure~\ref{fig:laminarbath}
for an example initial condition).  The bath atoms follow ballistic
motion~\eqref{ballistic},
and it is easily verified that the bath satisfies~\eqref{marge_vb} 
and~\eqref{invariance}. 
The bath atoms undergo elastic collisions~\eqref{collision2} with
the large particle.
\begin{figure}
\centerline{\input{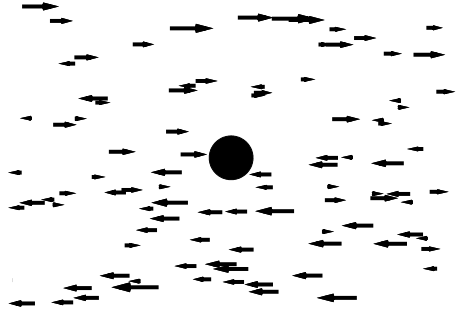}}
\caption{\label{fig:laminarbath}A large particle (of radius $R=2$) surrounded by
a laminar heat bath whose atoms are distributed according
to~\protect\eqref{fvq_shear}.  The velocity is centered around the
constant shear profile $A \qb,$ where $A$ is given
in~\protect\eqref{shear_mat}.  Since the bath atoms may only flow
parallel to the shear, the variance around the mean velocity along
each flow line $q_2=c$ is constant for all time.  We have chosen
$s=0.1$ $\beta=1,$ $m=1,$ and $\lambda=\frac{1}{16}.$ }
\end{figure}

Ignoring the effect of recollisions, we formally compute the limiting 
dynamics as $m\rightarrow 0$ to be
\begin{equation}
\label{ou_shear}
\begin{split} 
d \Qb &= \Vb dt, \\
M d \Vb &= - \gamma (\Vb - A \Qb) dt + \sigma d\Wb,
\end{split}
\end{equation}
where
\begin{equation*}
\gamma =  \frac{2 \sqrt{2\pi} \lambda R^2}{\sqrt{\beta}}
    \left( \begin{array}{ccc} 
            1 &0&0\\0&\frac{1}{2}&0\\0&0&\frac{1}{2}
           \end{array}\right), 
\qquad 
\sigma = \left[
       \frac{4 \sqrt{2\pi} \lambda R^2}{\sqrt{\beta^3}}
    \left( \begin{array}{ccc} 
            \frac{2}{3} &0&0\\0&\frac{1}{6}&0\\0&0&\frac{1}{6}
           \end{array}\right)          \right]^{\frac{1}{2}}.
\end{equation*} 
The details of this computation are given in Appendix~\ref{lam_gen}.
We note that for the laminar flow model considered here we have not
rigorously proven a convergence result and, in particular, we have not
justified ignoring recollisions.  On the other hand, a proof is
carried out for the model presented in Section~\ref{meas_dyn}.

The resulting dissipation and noise terms are both anisotropic, which
is due to the fact that the background laminar flow is itself
anisotropic.  These terms are larger in the direction of the flow, and
the dynamics do not satisfy a standard fluctuation-dissipation
relation in which $\frac{\beta}{2} \gamma^{-1} \sigma \sigma^T = 1,$
since  
\begin{equation}
\label{anisotropic}
\frac{\beta}{2} \gamma^{-1} \sigma \sigma^T = 
    \left( \begin{array}{ccc} 
            \frac{2}{3} &0&0\\0&\frac{1}{3}&0\\0&0&\frac{1}{3}
           \end{array}\right).
\end{equation} 
Furthermore, in the case where the shear flow is identically zero, $A =
0,$ the dynamics does not reduce to the standard Langevin
dynamics~\eqref{lang1} which was derived for the case of zero
background flow.  Since we desire a system of equations that
reduces to Langevin dynamics in the case $A=0$ (in particular because the Langevin dynamics is a widely accepted dynamics to model fluid flows at equilibrium), we consider a 
modification for removing the
anisotropy.

\subsubsection{Multiple laminar flows}
One may attempt to create an isotropic flow by
overlaying multiple background flows.  We consider the
superposition of three distinct laminar flows, each one with all
velocities aligned in a single direction, $\eb_i.$  Then we choose the
initial bath coordinates according to a Poisson field with distribution 
function
\begin{equation}
\label{fvq_multi_shear}
\begin{split}
d \mu_m(\qb,\vb) =& \left[ \frac{1}{3} \lambda_m Z^{-1} m^{1/2} \exp\left( 
 -\frac{\beta}{2} m(v_1 - s q_2)^2 \right) \delta(v_2) \delta(v_3),\right. \\
&+ \frac{1}{3} \lambda_m Z^{-1} m^{1/2} \delta(v_1) \exp\left( 
             -\frac{\beta}{2} m v_2^2 \right) \delta(v_3) \\
&\left. + \frac{1}{3} \lambda_m Z^{-1} m^{1/2} \delta(v_1) \delta(v_2) 
  \exp\left( -\frac{\beta}{2} m v_3^2 \right) \right] d \qb \, d \vb.
\end{split}
\end{equation}
Each atom moves in one of the three coordinate directions, since the
presence of the delta distributions means that at most one of the
velocity coordinates is nonzero.  Since the bath atoms do not have any
self-interaction, the different flows do not `mix' in any way, and in
particular this bath measure is invariant under ballistic motion.
Figure~\ref{fig:multiplebath} depicts a realization of this
bath in the two-dimensional case.  
\begin{figure}
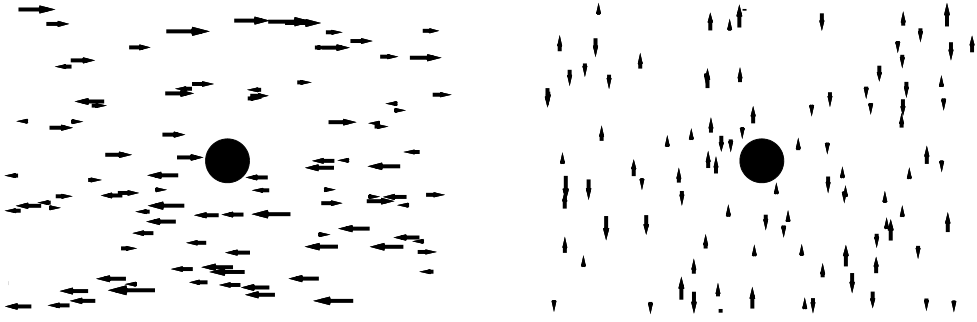

\centerline{\input{figs/bath2}  \input{figs/bath3} }
\caption{\label{fig:multiplebath} Two components of
the laminar bath which are superimposed to create an environment
for the particle consistent with the measure~\protect\eqref{fvq_multi_shear}.
We have chosen $s=0.1,$ $\beta=1,$ $m=1,$ and $\lambda=\frac{1}{16}.$}
\end{figure}

According to the formal calculations carried out in 
Appendix~\ref{lam_gen}, where we again ignore the effect of
recollisions, we find that the limiting dynamics of the large
particle is
\begin{equation} 
\label{ou_shear2}
\begin{split}
d\Qb &= \Vb dt, \\
M d\Vb &= - \gamma \left(\Vb - \frac{1}{2} A \Qb \right) dt 
	 + \sigma d\Wb,
\end{split}
\end{equation} 
where
\begin{equation*}
\gamma =  \frac{4 \sqrt{2\pi} \lambda R^2}{\sqrt{\beta}},
\qquad 
\sigma =  \left[ \frac{4 \sqrt{2\pi} \lambda R^2}{\sqrt{\beta^3}}
\right]^{1/2}.
\end{equation*}
As desired, the dynamics in~\eqref{ou_shear2} is no longer an anisotropic equation in
contrast to~\eqref{ou_shear} since $\gamma$ and $\sigma$ are scalar quantities.  We note that the damping term on the large particle is
relative to {\em half} the velocity of the background shear flow and 
that the fluctuation-dissipation relation is not satisfied, since 
$\frac{\beta}{2} \gamma^{-1} \sigma^2 = \frac{1}{2} \neq 1.$  
While we have chosen the inverse temperature for each bath equal to $\beta$ in equation~\eqref{fvq_multi_shear}, we find in Appendix~\ref{lam_gen} 
that there is no choice of inverse
temperatures for the three laminar bath components that allows us to
construct
an isotropic damping in which the damping term is relative to the 
full background flow.  We note that this issue is not specific 
to the three dimensional case considered here.  

We do not pursue this model further for two reasons.  First and foremost,
it is not clear how to generalize these laminar bath flow models 
from simple shear flows to more general incompressible flows.  Second, a heat bath with
temperature proportional to $\beta^{-1}$ and strain rate $A$ gives a 
macroscopic dynamics whose temperature is proportional to 
$\frac{1}{2} \beta^{-1}$ and whose damping is relative to $\frac{1}{2} A.$ 
We expect damping relative to $A$ as in~\eqref{ou_flow_intro}, which agrees
with NEMD dynamics such as g-SLLOD (see Section~\ref{sec:gsllod} for a description
of the g-SLLOD dynamics and their relation to~\eqref{ou_flow_intro}). 
While it would be possible to choose bath parameters $A_{\rm bath} = 2 A$ 
and $\beta_{\rm bath} = \frac{1}{2} \beta$ to give 
the desired macroscopic parameters, $A$ and $\beta,$ to the large particle, 
we instead consider a different 
approach which both handles general incompressible flows and shows agreement
between the bath parameters and the parameters of the derived dynamics.

\subsection{Non-Hamiltonian bath dynamics}
\label{meas_dyn}

In this section we describe an approach based on modifying the bath
dynamics to no longer follow ballistic trajectories~\eqref{ballistic}.
The particle velocity accelerates so that the
distribution of velocities relative to the background flow is preserved.

\subsubsection{Model and convergence results} 
We consider a bath whose initial condition is characterized by
the measure
\begin{equation}
\label{bath_meas2}
d \mu_m(\qb, \vb) = \lambda_m f_m(\qb, \vb) \, d \qb \, d \vb \qquad
\text{for } \qb,\vb \in \RR^d,
\end{equation}
where $f_m$ has the form
\begin{equation}
\label{fvq}
f_m(\qb,\vb) = m^{d/2} f(m^{1/2}(\vb - A \qb))
\end{equation} 
for $A \in \RR^{d \times d}$ (Figure~\ref{fig:genbath}) and where $\lambda_m = m^{-1/2} \lambda$.  As in
Section~\ref{mechdef}, we assume that $f(\vb)$ is a rotationally
invariant probability distribution, and that the first four moments are finite.   Any
distribution of the form~\eqref{bath_meas2} with~\eqref{fvq}
satisfies~\eqref{marge_vb}.  We additionally assume that $f$  
is a decreasing function of $|\vb|.$  This additional assumption is
used below when approximating the trajectory of the mechanical process
with a Markov process in
Appendix~\ref{me_to_ma}.
\begin{figure}
\centerline{\input{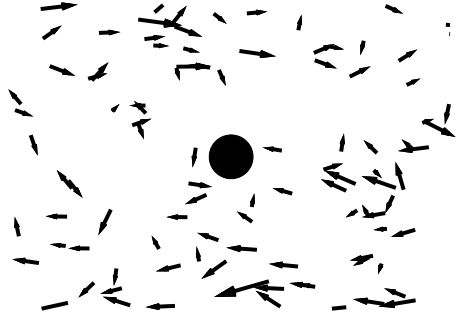}}
\caption{\label{fig:genbath}A large particle (of radius $R=2$) immersed in 
a heat bath whose atoms are distributed according to~\protect \eqref{fvq}. 
The velocities have Gaussian distribution centered around a
shear flow, with $A$ given in~\protect \eqref{shear_mat}.  
We have chosen $s=0.1,$ $\beta=1,$ $m=1,$ and $\lambda=\frac{1}{16}.$}
\end{figure}

In order to satisfy~\eqref{invariance}, we now change the underlying
bath dynamics and consider the following non-Hamiltonian dynamics for
the bath atoms: 
\begin{equation}
\label{nonHam1}
\begin{split}
d \qb &= \vb dt, \\
d \vb &= A \vb dt.\\
\end{split}
\end{equation} 
It is still the case
that bath atoms do not interact with one another, but now
the bath atoms accelerate in a way that leaves the desired velocity
distribution~\eqref{bath_meas2}-\eqref{fvq} invariant.
The corresponding Liouville equation is 
\begin{equation*}
\begin{split} 
\label{gen_arg2}
\partial_t \rho(\qb, \vb,t) 
  &= - \nabla_\qb \cdot (\vb \rho(\qb, \vb,t)) 
- \nabla_\vb \cdot (A \vb \rho(\qb, \vb,t)).
\end{split} 
\end{equation*}
As long as $\tr A = 0,$ which is equivalent to having an
incompressible background flow, any function of the form $\rho(\qb, \vb, t) =
f(\vb-A\qb)$ is invariant under the dynamics.  This invariance
is our motivation
for choosing the dynamics~\eqref{nonHam1}, and in particular, we are not
restricted to shear flows.  We interpret the dynamics by considering
the relative velocity 
\begin{equation} 
\label{relvel}
\vcen = \vb - A \qb.
\end{equation}   
In terms of $(\qb, \vcen),$ the dynamics~\eqref{nonHam1} is
\begin{equation}
\label{nonHam2}
\begin{split}
d \qb &= (A \qb + \vcen) dt, \\
d \vcen &= 0.\\
\end{split}
\end{equation} 
Thus the relative velocity, $\vcen,$ does not change in time and the
velocity of the atoms is at all times equal to the background flow
plus an initial perturbation of thermal origin.  However, the choice of dynamics~\eqref{nonHam1} is motivated by the fact that it
satisfies~\eqref{invariance}, rather than by the above physical
interpretation.

As before, the large particle undergoes elastic collisions with the 
bath atoms according to the rule~\eqref{collision2}. In the limit 
$m\rightarrow 0,$ the particle dynamics
converges to the nonequilibrium Langevin dynamics~\eqref{ou_flow_intro}
which we recall here:
\begin{equation}
\begin{split}
\label{ou_flow}
d \Qb &= \Vb dt, \\
M d \Vb &= - \gamma (\Vb-A \Qb) dt + \sigma d\Wb,
\end{split}
\end{equation}
where, as in the zero background flow case~\eqref{lang_coeff},  
\begin{equation}
\label{neld_coeff}
\gamma = \frac{4 \lambda R^{d-1} S_{d-1}}{d} \Phi_1, 
  \qquad 
\sigma = \left[\frac{4 \lambda R^{d-1} S_{d-1}}{d} \Phi_3
\right]^{\frac{1}{2}},
\end{equation} 
where $S_{d-1}$ denotes the surface area of the $(d-1)$-sphere, and
$\Phi_i$ is defined in~\eqref{phi_i_2}.  As in the case of zero background
flow in Section~\ref{mechdef}, this dynamics satisfies a
standard fluctuation-dissipation relation~\eqref{fdr}, with temperature
equal to the bath's temperature, provided that 
\begin{equation*}
\frac{\Phi_1}{\Phi_3} = \frac{\beta}{2},
\end{equation*}
which holds true for the Gaussian distribution $f(\vb) = Z^{-1} \exp(-\frac{\beta}{2} |\vb|^2),$
for example. 
We establish the following result in the
Appendices~\ref{sec:bath_gen}-\ref{wellposed}.
\begin{theorem}
\label{thm:main}
Consider the mechanical process $(\Qb_m, \Vb_m),$ defined
in~\eqref{particle_dynamics},~\eqref{collision2},
and~\eqref{nonHam1} with bath measure defined in~\eqref{bath_meas2}
and~\eqref{fvq}, and let $(\Qb, \Vb)$ denote the solution to the
nonequilibrium Langevin equations~\eqref{ou_flow} with
coefficients~\eqref{neld_coeff}.  For any $T>0,$ 
\begin{equation*}
(\Qb_m, \Vb_m)_\tint \lconv (\Qb, \Vb)_\tint
\end{equation*}
as $m \rightarrow 0,$ where convergence is with respect to 
the Skorokhod metric~\eqref{skor} on $\D([0,T])$.
\end{theorem}

The structure of the proof is the same as in~\cite{durr80}: we use a
Markov approximation to the mechanical process and split the
convergence proof into two steps.  In Appendix~\ref{sec:bath_gen}, we
define the Markov process that approximates the mechanical process and
outline in more detail the convergence proof.  In
Appendix~\ref{ma_to_ou}, we show that the Markov process converges in
law to the solution of the nonequilibrium Langevin
dynamics~\eqref{ou_flow}. In Appendices~\ref{norecollide}
and~\ref{me_to_ma} we show that the difference between the Markov
process and the mechanical process converges in probability to zero.
These two convergence results are enough to deduce that the mechanical
process converges in law to the nonequilibrium Langevin
dynamics~\cite{bill99}.  In Appendix~\ref{wellposed}, we show that the
mechanical process is well-defined for
almost every initial condition up to a positive stopping time.

\subsubsection{Relationship to g-SLLOD}
\label{sec:gsllod}
The g-SLLOD equations of motion~\cite{tuck97,evan07} (also called
p-SLLOD~\cite{edwa06}) are given by 
\begin{equation}
\label{gsllod}
\begin{split}
d \Qb_i &= (A \Qb_i + \Vcen_i) dt, \\
M d \Vcen_i &= (- \nabla_{\Qb_i} E(\Qb) - M A \Vcen_i - M A A \Qb_i) dt. \\
\end{split} 
\end{equation}
Note that these equations are written in terms of relative velocity, 
$\Vcen_i = \Vb_i - A \Qb_i.$
This system of equations is used to simulate a molecular system in a
non-zero background flow.  The particles interact through the potential field
$E(\Qb).$
The g-SLLOD equations are simply a change of variables from the Newton
equations of motion in the reference frame
\begin{equation*}
\begin{split}
d \Qb_i &= \Vb_i \, dt, \\
M d \Vb_i &= - \nabla_{\Qb_i} E(\Qb) \, dt, \\
\end{split} 
\end{equation*}
to the relative velocity.  If we add a
Langevin damping term and noise term to the dynamics
in~\eqref{gsllod}, we have
\begin{equation*}
\begin{split}
d\Qb_i  &= (\Vcen_i + A \Qb_i) dt, \\
M d\Vcen_i &= (- \nabla_{\Qb_i} E(\Qb) - M A \Vcen_i - M A A \Qb_i) dt 
             - \gamma \Vcen_i \, dt + \sigma d\Wb. 
\end{split} 
\end{equation*}
Transforming back to the reference frame, this gives
\begin{equation}
\label{gsllod_T}
\begin{split}
d\Qb_i  &= \Vb_i  dt, \\
M d\Vb_i &= - \nabla_{\Qb_i} E(\Qb) dt - \gamma (\Vb_i - A \Qb_i) dt 
+ \sigma d\Wb, 
\end{split} 
\end{equation}
which is~\eqref{ou_flow}, with the addition of an external
potential $E$.
Thus, our construction is consistent with the application of a Langevin
thermostat to the g-SLLOD equations of motion and thus provides
a {\em{derivation}} of the g-SLLOD equations from a heat bath model.

\subsubsection{Generalizations}
\label{sec:gen}

We have chosen the microscopic dynamics~\eqref{particle_dynamics} for
the large particle and~\eqref{nonHam1} for the bath atoms.   
First, we may make a more general choice of
microscopic dynamics for the large particle between collisions,
\begin{equation}
\label{gen_particle_dynamics}
\begin{split} 
d\Qb_m &= \Vb_m \, dt, \\
M d\Vb_m &= F(\Qb_m, \Vb_m) dt,
\end{split}
\end{equation}
replacing~\eqref{particle_dynamics} with~\eqref{gen_particle_dynamics}
in the definition of the mechanical process.
We assume $F$ is uniformly Lipschitz continuous on $\RR^d \times \RR^d$
and $\nabla_{\Vb} \cdot F(\Qb,\Vb) = 0$ 
which implies that both the mechanical process for finite $m$ 
and the limiting SDE are  
well-posed (the proof given in Appendix~\ref{wellposed} for the $F=0$
case depends on the fact that the flow map of the mechanical process
preserves Lebesgue measure). 
The limiting dynamics of the large particle becomes 
\begin{equation}
\begin{split} 
d\Qb &= \Vb dt, \\
M d\Vb &= F(\Qb, \Vb) dt - \gamma (\Vb-A \Qb) dt + \sigma d\Wb.
\end{split}
\end{equation}
The proof of this result is a straightforward extension of the
analysis given here, where the Lipschitz continuity allows us to
bound the particle trajectory over finite time intervals.

For the specific case of $F(\Qb,\Vb)=M A \Vb$ with $\tr A =0,$ 
which is the same acceleration
term as in the microscopic
dynamics chosen for the bath, we find the stochastic dynamics
\begin{equation}
\label{ou_A_flow}
\begin{split} 
d\Qb &= \Vb dt, \\
M d\Vb &= M A \Vb dt- \gamma (\Vb-A \Qb) dt + \sigma d\Wb.
\end{split}
\end{equation}
This dynamics was obtained in~\cite{mcph01} for the case of shear 
flow by applying the Kac-Zwanzig formalism, thereby obtaining a generalized
Langevin dynamics, and assuming that the memory
kernel converges to a delta function (that is, there is no
 memory in the system).
The dynamics~\eqref{ou_A_flow} is of particular interest since we 
can explicitly check that 
$\rho(\Qb, \Vb) = \exp\left( -\frac{\beta}{2} |\Vb -A \Qb|^2   \right)$
is an invariant measure of~\eqref{ou_A_flow}.  In
contrast we do not know an analytical expression for stationary states
of~\eqref{ou_flow} for
general choice of $A.$

Second, we could consider the case of a non-homogeneous incompressible 
background flow, where
the average velocity at the point $\qb^*$ in~\eqref{marge_vb} is $\ub(\qb^*)$ where 
$\ub$ is a divergence-free vector field.
For any such $\ub,$ the microscopic bath dynamics 
\begin{equation}
\label{nonHam_gen}
\begin{split}
d \qb &= \vb dt, \\
d \vb &= \nabla \ub(\qb) \vb dt,\\
\end{split}
\end{equation} 
preserves any bath measure of the form~\eqref{bath_meas2}-\eqref{fvq}
with  $\vb - A \qb$ replaced by $\vb - \ub(\qb).$
In this case, the dynamics of the large particle in the 
mechanical process converges, in the
limit $m \rightarrow 0,$ to 
\begin{equation}
\begin{split} 
d\Qb &= \Vb dt, \\
M d\Vb &= - \gamma (\Vb- \ub(\Qb)) dt + \sigma d\Wb.
\end{split}
\end{equation}
This extension likewise does not pose any theoretical
difficulties, at least when $\ub$ is Lipschitz continuous.
This
reduces to the case presented above when $\ub(\qb) = A \qb,$
where $\divop \ub(\qb) = 0$ is equivalent to $\tr A = 0.$

Finally, a much more challenging extension is the case of {\em
{multiple}}
large particles.  The restriction to a single large particle is a necessary
assumption for the argument here and in~\cite{durr80,durr83,cald89},
since it allows us to estimate the distribution of velocities of the
bath atoms that collide with the particle.  In particular, fast moving
bath atoms can collide at most once with the large particle, whereas
in the multi-particle case these atoms could bounce between large
particles and possibly collide with them many times.  Extending the
proof of Theorems~\ref{thm:dlg} and~\ref{thm:main} to the case of
multiple particles immersed in a bath is a non-trivial task, as atoms
of any speed can bounce between the large particles, and there is also
a shadowing effect that particles have on one another.
Kotelenez~\cite{kote08} treated the multi-particle case using a 
mean field interaction representing a mesoscale interaction and showed
how the bath can generate close-range forces among the large particles.  
Kusuoka
and Liang treated the multi-particle case as well, using a weaker
bath-particle interaction which prevents atoms from bouncing back and
forth among the large particles~\cite{kusu06}.  We do not pursue a
multi-particle derivation in this work; however, we do perform
numerical tests of a natural extension of the limiting
equations~\eqref{ou_flow} to the multi-particle case in
Section~\ref{sec:num}.

\section{Numerics for the multi-particle case}
\label{sec:num}

We consider the application of the derived equations~\eqref{ou_flow}
to the simulation of a fluid composed of many identical large
particles.  We note that the derivation for Theorem~\ref{thm:main}
only applies to the case of a {\em{single}} particle, with no forcing
except through the bath, whereas we now apply the equation to a system
with {\em{many}}
large particles that interact with one another.  The many-particle
case is the case of interest in applications, and we numerically
investigate the qualitative behavior of the proposed dynamics to test
its suitability for sampling flows with a given mean behavior.  We
consider a system in $\RR^3$ consisting of $N$ large particles whose
position and velocity $(\Qb, \Vb) \in \RR^{6 N}$ evolve according to
the dynamics~\eqref{gsllod_T} which we recall here:
\begin{equation}
\label{neld_many}
\begin{split}
d \Qb_i &= \Vb_i \, dt, \\
M d \Vb_i &= - \nabla_{\Qb_i} E(\Qb) \, dt  - \gamma (\Vb_i-A \Qb_i) dt + \sigma
d\Wb_i.
\end{split}
\end{equation}
The index $i = 1,\dots,N$ runs over all particles, whereas the lack
of index in the argument of $E$ denotes the fact that it is a function
of the full vector
$\Qb \in \RR^{3N}$. 
The Fokker-Planck equation for~\eqref{neld_many} is
\begin{equation}
\label{fp_main}
\begin{split}
\partial_t \rho =& -\nabla_\Qb \cdot (\Vb \rho) 
+ M^{-1} \nabla_\Vb \cdot (\nabla E(\Qb) \rho)  %\\  &
   + M^{-1} \nabla_\Vb \cdot [\gamma(\Vb - A \Qb)\rho\,]
+ \frac{1}{2}  M^{-2} \sigma^2 \Delta_{\Vb} \rho.
\end{split}  
\end{equation}  
We do not have an analytic expression for solutions of the
time-dependent equation~\eqref{fp_main} or any steady-state solutions.  
Thus, it is useful to carry out numerical experiments
on the response of the multi-particle system to the background forcing
and explore the resulting constitutive relation.

\subsection{Simple Lennard-Jones fluid}

We now run numerical tests on a 3D flow of particles with
Lennard-Jones interactions and background shear flow 
$$A = \left[
\begin{array}{ccc} 0&s&0 \\ 0&0&0 \\ 0&0&0 \end{array} \right].$$
A closely related dynamics, differing in the choice of boundary conditions 
(we detail our choice below) and how the external forcing is handled, is 
considered in~\cite{joub11} where the authors perform
rigorous asymptotic analysis on the invariant measure as well as
numerical viscosity experiments.

We let $\Omega = [-L/2, L/2]^3$ be the computational domain, and we
initially arrange $N$ particles on a uniform cubic lattice with
spacing $a,$ where $L = a N^{1/3}.$  The initial velocities are random, drawn
from the measure $Z^{-1} \exp(-\frac{\beta}{2} |\Vb - A \Qb|^2).$  We
note that over the long times that we simulate, the solution is
insensitive to the choice of initial conditions (we have also tested
with zero initial velocities).  The potential $E$ denotes the
Lennard-Jones potential energy with cutoff
$$
E(\Qb) = \sum_{i=1}^{N} \sum_{\substack{j=i+1}}^{N} \phi(r_{ij}),
$$
where 
\begin{equation*}
\phi(r) = 
\begin{cases}
\displaystyle 4 \eps \left(\frac{1}{r^{12}} - \frac{1}{r^6}
\right) + c_1 r + c_2& \textrm{ if } r < r_{\rm cut},  \\
0 & \textrm{ if } r \geq r_{\rm cut}.\\
\end{cases}
\end{equation*} 
The constants $c_1$ and $c_2$ are chosen so that $\phi(r)$ and
$\phi'(r)$ are continuous at $r_{\rm cut}.$ 
The distance $r_{ij} =|\Qb_i - \Qb_j|$ is computed taking into 
account the boundary conditions which we now make precise.  

We apply the Lees-Edwards
boundary conditions~\cite{alle89,evan07} to the system, which
generalize standard periodic boundary conditions to be
consistent with shear flow.  
At time $t$, the particle 
$(\Qb(t), \Vb(t))$
has periodic replicas 
\begin{equation*}
\begin{split}
(Q_1(t), &Q_2(t), Q_3(t), V_1(t), V_2(t), V_3(t))\\
&= (Q_1(t) + m L + t n s L, Q_2(t) + n L, Q_3(t) + k L, 
                 V_1(t) + n sL, V_2(t), V_3(t))
\end{split}
\end{equation*} 
for all $m,n,k \in \ZZ.$  This ensures that the periodic replicas of the
system move consistently with the shear.

\begin{figure}
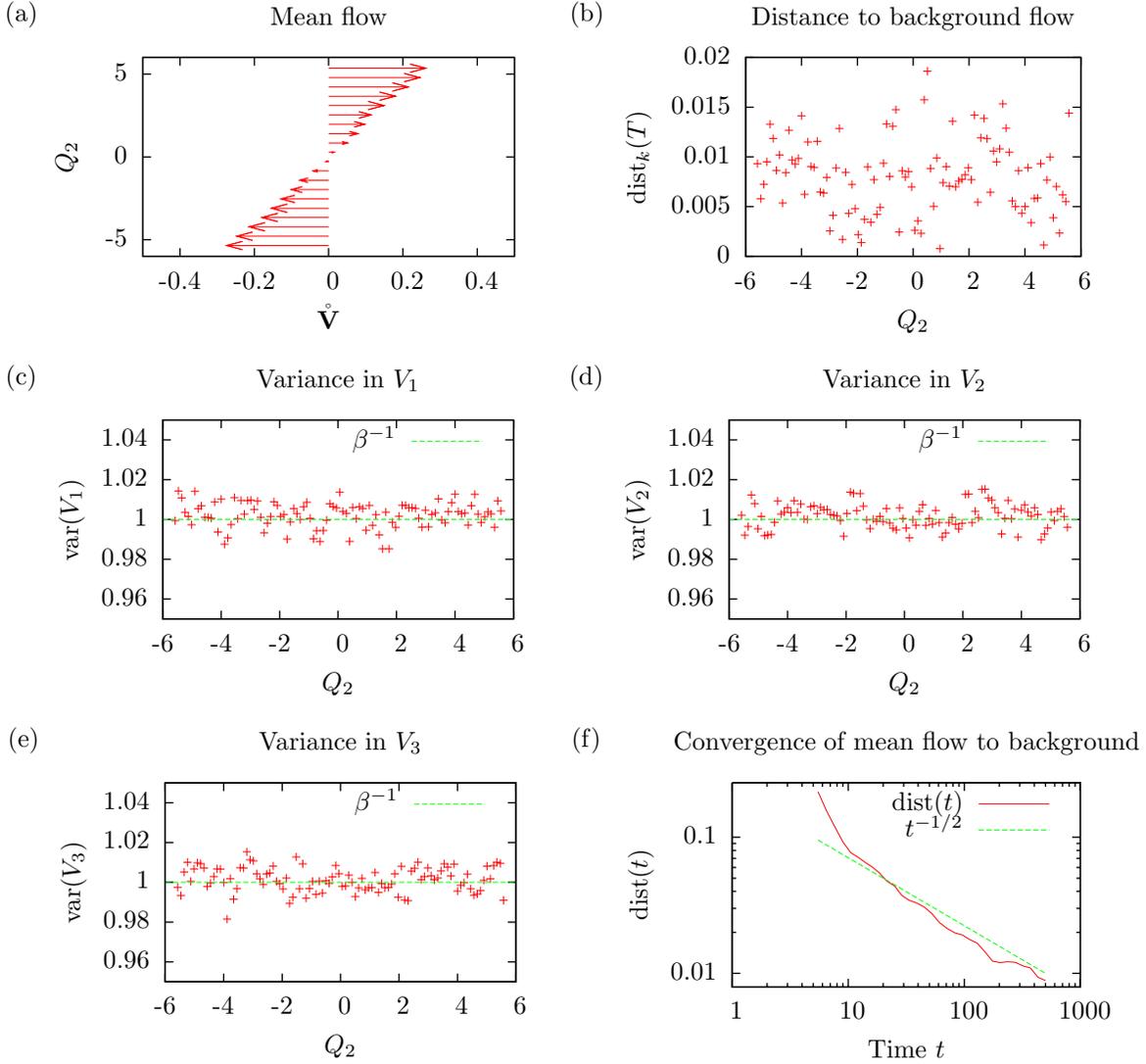

\centerline{ 
(a) {\raisebox{-1.7in}{\input{figs/vel_1_shear.tex}}}
(b) {\raisebox{-1.7in}{\input{figs/vel_1_error.tex}}} } 
\centerline{ 
(c) {\raisebox{-1.7in}{\input{figs/var_1_x.tex}}}
(d) {\raisebox{-1.7in}{\input{figs/var_1_y.tex}}} }
\centerline{
(e) {\raisebox{-1.7in}{\input{figs/var_1_z.tex}}} 
(f) {\raisebox{-1.7in}{\input{figs/vel_1_evt.tex}}}  } 
\centerline{\parbox{16cm}{\caption{\label{fig:vel}
(a) The mean velocity~\protect \eqref{mv_avg} of the particles within
horizontal slices~\protect \eqref{slice} for a shear flow with strain
rate $s=0.05$, averaged in $t$.  We plot 20 of the 100 bins used in
creating the statistics.  
(b) The distance~\protect \eqref{rel_err} between the computed mean velocity and
the background flow  is plotted versus $Q_2$ at the final time
$t_\ell= T=500.$ 
(c), (d), and (e) The variance of particle velocity around the mean is
plotted versus $Q_2,$ where we again average in $Q_1,Q_3,$ and $t.$  This is well centered around $\beta^{-1}.$  
(f) The distance~\protect \eqref{mv_err} of the computed mean velocity to the
background flow is plotted
versus time in a $\log$-$\log$ axis.  We observe the expected
O($t^{-1/2}$) decay of the distance.  
}}}
\end{figure}
We first perform a single run with the parameters
$\eps = 1, N = 1000, \beta = 1.0, s=0.05, M=1,
a={0.7}^{-1/3}, r_{\rm cut}
= 2.6, \gamma = 1, \sigma = \sqrt{2 \gamma \beta^{-1}}.$  For this
choice of parameters, the Lennard-Jones particles are in a
fluid regime (see e.g.~\cite{rowl97}).  We run the simulation up to time
$T = 500,$ with a stepsize of $\Delta t = 0.005$ using the following
splitting algorithm.  We let $(\Qb^n, \Vb^n)$ denote the approximate
particle position and velocity at time $t= n \Delta t.$  We let
$\alpha = e^{-\frac{\gamma}{M} \Delta t},$ and for $i = 1,\dots, N$ we define
$(\Qb_i^{n+1}, \Vb_i^{n+1})$
by
\begin{equation*}
\begin{split} 
\Vb_i^{n+1/2} &= \Vb_i^{n} - \frac{\Delta t}{2} M^{-1} 
                             \nabla_{\Qb_i} E (\Qb^n)\\
\Qb_i^{n+1} &= \Qb_i^n + \Delta t \Vb_i^{n+1/2} \\
\Vb_i^{*} &= \Vb_i^{n+1/2} - \frac{\Delta t}{2} M^{-1} 
                             \nabla_{\Qb_i} E (\Qb^{n+1})\\
\Vb_i^{n+1} &= \alpha \Vb_i^{*} + (1-\alpha) A \Qb_i^{n+1} +
       \left(\frac{1-\alpha^2}{\beta M} \right)^{1/2} \Gb_i^{n} \\
\end{split} 
\end{equation*}
where $\Gb_i^n \sim \mathcal{N}(0,I)$ where $I$ denotes the $d \times d$
identity matrix.
This is a composition of a Verlet step applied to the Hamiltonian
portion
\begin{equation}
\begin{split}
d \Qb_i &= \Vb_i dt, \\
M d \Vb_i &= - \nabla_{\Qb_i} E(\Qb) \, dt
\end{split}
\end{equation}
with an exact integration of the remaining terms which
represent the effect of the heat bath on the large particles,
\begin{equation}
\begin{split}
d \Qb_i &= 0, \\
M d \Vb_i &= - \gamma (\Vb_i-A \Qb_i) dt + \sigma d\Wb_i .
\end{split}
\end{equation}
 
\subsection{Mean flow and stress}
In Figure~\ref{fig:vel} we display the results of our numerical 
simulation of shear flow.  We partition the domain into $K=100$
slices, 
\begin{equation}
\label{slice}
R_k=[-L/2,L/2] \times \left[-L/2+\frac{k-1}{K} L,-L/2+\frac{k}{K}
L\right] 
\times [-L/2,L/2] \end{equation}
for $k=1,\dots,K.$
At any time $t_\ell = \ell \Delta t,$ we define a time average of the mean velocity in each
slice $k$ by 
\begin{equation}
\label{mv_avg}
\Vavgb(t_\ell, k) = 
\frac{\sum_{n=0}^{\ell} \sum_{i=1}^{N} \Vb^n_i \one_{R_k}(\Qb_i^n)}
     {\sum_{n=0}^{\ell} \sum_{i=1}^{N} \one_{R_k}(\Qb_i^n)}
\end{equation}
where $\Vb^n_i$ denotes the velocity of particle $i$
at time $n \Delta t$ and where $\one_{R_k}$ denotes the indicator
function of the slice $R_k.$  Since the flow is uniform in all but the
$\eb_2$-direction, averaging on slices  
helps improve the convergence of the flow statistics.  The time average
over the full simulation of the mean
velocity is plotted in Figure~\ref{fig:vel}(a), and we observe a linear
profile equal to the applied background flow.

We define the distance of the mean flow in each slice to the
background flow,
\begin{equation}
\label{rel_err}
\begin{split}
\mathrm{dist}(t_\ell,k) 
= \left((\Vavg_{1}(t_\ell,k) - s Q_{2}(k))^2+
\Vavg_{2}^2(t_\ell,k) + \Vavg_{3}^2(t_\ell,k)
\right)^{1/2},
\end{split}
\end{equation}
where $Q_2(k) = -L/2 + \frac{k-1/2}{K} L$ is the $y$-coordinate of the
center of rectangle $R_k.$ We plot the distance as a function of $Q_2$
in Figure~\ref{fig:vel}(b) and we observe that it is uniform throughout
the domain.  The variance of the velocity is computed for each slice
and is displayed in Figure~\ref{fig:vel}(c), (d), and (e).  The
variance closely matches the expected variance due to the background
temperature, $\frac{1}{k_B \beta} = 1.$ In particular, it is
statistically the same in each direction so that we observe a scalar
temperature in contrast to the laminar flow case
of~\eqref{anisotropic}.  In Figure~\ref{fig:vel}(f), we plot  the
distance of the mean flow $\Vavgb(t_\ell,k)$ to the background flow
versus time, 
\begin{equation}
\label{mv_err}
\mathrm{dist}(t_\ell) = \left(\frac{1}{K} 
\sum_{k=1}^{K} \left(\Vavg_{k,1}(t_\ell) - s Q_{k,2}\right)^2 +
\Vavg_{k,2}^2+ \Vavg_{k,3}^2
\right)^{1/2}.
\end{equation}
We observe an O($t^{-1/2}$) decay in the distance, which is the
convergence rate to the mean expected for Monte Carlo empirical
averages.

\begin{figure}
\centerline{ \input{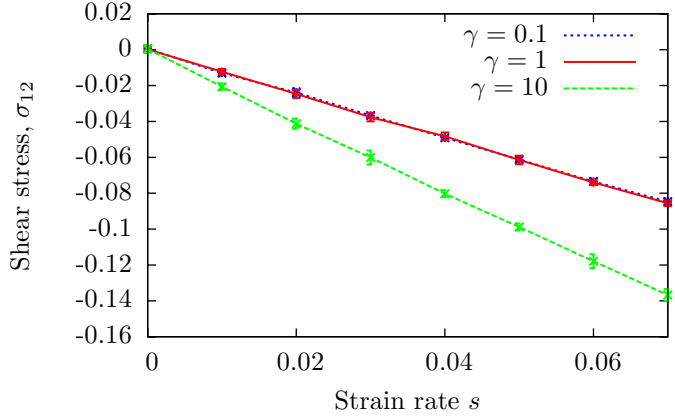}}
\caption{\label{fig:stress}Shear stress $\sigma_{12}$ 
vs velocity gradient.  We observe that the viscosity increases with the heat bath
parameter $\gamma,$ though for the two smaller values of $\gamma,$ the viscosity is 
extremely close.  For the values $\gamma=0.1,$ $\gamma=1.0,$ and
$\gamma=10,$ we find viscosities $\eta= 1.2175,$ $\eta = 1.2254,$ and
$\eta = 1.9518$, respectively.}
\end{figure}
In Figure~\ref{fig:stress}, we display the $\sigma_{12}$ term of
the 
shear stress, computed using the Irving-Kirkwood
formula~\cite{irvi50} which we average over all time steps 
\begin{equation}
\label{irv_kirk}
\boldsymbol{\sigma} = \frac{1}{|\Omega|} \sum_{i =1}^{N} 
\left(M (\Vb_i^{n}- A \Qb_i^n) \otimes
(\Vb_i^{n}- A \Qb_i^n) + \frac{1}{2} \sum_{\substack{j = 1\\j \neq i}}^{N} ( \Qb_i - \Qb_j)
\otimes f^{(i j)}\right)
\end{equation}
where 
\begin{equation*}
f^{(i j)} = - \frac{\phi'(|\Qb_i - \Qb_j|) (\Qb_i - \Qb_j)}{|\Qb_i - \Qb_j|}
\end{equation*} 
denotes the force on atom $i$ due to atom $j$ and $|\Omega| = L^3$ is
the volume of $\Omega.$  We note that it is standard to subtract the
local average velocity from the pressure term~\cite{irvi50}: hence we
have subtracted the background flow at each step.  We run ten 
independent simulations 
for each choice of the strain rate
$s=0, 0.01, 0.02, \dots, 0.07$ and the parameter $\gamma =
0.1, 1.0, 10.$ 
The simulations are all run to time
$T=500,$ with $N,$
$\beta,$ $M,$ $a,$ $r_{\rm cut},$ and $\sigma$ chosen as above.   We
plot the stress along with error bars equal to 3 times the standard
deviation.  Time averages of the shear stress $\sigma_{12}$ are plotted
against the strain rate $s.$  We notice
a linear relation, and define the shear viscosity $\eta =
-\frac{\sigma_{12}}{s}.$  Fitting the data, we find viscosity (with
standard deviation) $\eta=1.2175 \pm 0.0085$ for $\gamma=0.1$ whose
confidence interval overlaps with the value $\eta = 1.142 \pm 0.087$
reported in~\cite[Table IV]{rowl97}.  
The algorithm successfully simulates a system out of equilibrium, and
the computed viscosity is consistent with previous computations for
Lennard-Jones fluids.

\appendix

\section{Convergence in the small mass limit}
\label{sec:bath_gen}

In this section, we outline the proof of the convergence of the
mechanical process $(\Qb_m, \Vb_m)$ to the solution of the
nonequilibrium Langevin dynamics~\eqref{ou_flow}, as stated in
Theorem~\ref{thm:main}.  We construct a Markov approximation
$(\Qtb_m,\Vtb_m)$ that only counts the `fast collisions' of the
mechanical process (which we define below).  
The Markov approximation acts as an intermediate process between
the
mechanical process and the nonequilibrium Langevin dynamics.
We prove in Appendix~\ref{ma_to_ou} that the Markov approximation
$(\Qtb_m,\Vtb_m)$ converges to the solution of the nonequilibrium
Langevin dynamics~\eqref{ou_flow} in the sense that
\begin{equation}
\label{lconv}
(\Qtb_m,\Vtb_m)_\tint \lconv (\Qb, \Vb)_\tint  \textrm{ as  } m \rightarrow 0.
\end{equation} 
In Appendix~\ref{me_to_ma} we show that the mechanical process and
Markov approximation are close to each other, in the sense that
\begin{equation}
\label{pconv}
(\Qtb_m, \Vtb_m)_\tint  - (\Qb_m, \Vb_m)_\tint  \pconv 0 
	  \textrm{ as } m \rightarrow 0. 
\end{equation}  
This is shown by coupling the Markov process to the mechanical
process, that is for each realization of the mechanical process we  
associate a (random) set of realizations of the Markov process
that are (with high probability) close to the realization of the
mechanical process.  By a standard result in probability
theory~\cite[Theorem 3.1]{bill99}, properties~\eqref{lconv}
and~\eqref{pconv} 
allow us to conclude Theorem~\ref{thm:main}, that is, that 
\begin{equation}
(\Qb_m, \Vb_m)_\tint \lconv (\Qb, \Vb)_\tint 
\textrm{ as } m \rightarrow 0. 
\end{equation}

The proof here is structured as in~\cite{durr80}, with changes made to
handle the fact that the distribution of the initial bath atom velocities 
depend on space.
We work with the pairing $(\Qb_m, \Vb_m),$ since in our case, we
cannot create a Markov approximation for the velocity alone.  
We give a self-contained proof, while also pointing out where
it differs from the original argument in~\cite{durr80}.

\subsection{Rate of fast collisions}

As noted in Remark~\ref{rem:notMarkov}, the mechanical process is not
a Markov process itself since a bath atom may collide more than once
with the large particle or certain collisions may be impossible due to
past motion of the particle.  However, the expected value of the
relative speed of a bath atom is  
$O(m^{-1/2}),$ which is large since $m$ is assumed small, whereas the
expected value of the relative speed of the large particle is much
smaller, $O(1)$.
Most collisions happen between a fast-moving bath atom and a 
slower-moving large particle, and we show below that
for such collisions there is no chance of recollision between the
large particle and the bath atom involved.  This motivates the
following definition of ``fast collisions'' and the introduction below of a
stopping time on the trajectory of the mechanical process that controls
the large particle's velocity and position.

\begin{definition}
\label{fast}
We call a collision a {\emph{fast}} collision if the
normal component of the relative velocity~\eqref{relvel} of the bath
atom before the collision occurs satisfies 
\begin{equation*}
\vce_{\rm n} > c_m := m^{-1/5}
\end{equation*}  
where $\vce_{\rm n} = \vcen \cdot \eb_{\rm n}.$
Every other collision is called a {\emph{slow}} collision.
\end{definition}
The particular scaling of $c_m$ is chosen so
that $c_m\rightarrow \infty$ and  $c_m^2 m^{1/2} \rightarrow 0$ as $m
\rightarrow 0.$  We use the second limit in Appendix~\ref{sec:slow} to
bound the total effect of slow collisions.

\begin{definition}
Let $T>0$ be given.  For a given realization of the mechanical process
$(\Qb_m, \Vb_m),$ we define the stopping time 
\begin{equation}
\label{stop_time_mech}
\tau_m = \min\left( \inf_{t \geq 0} \left\{  t : \|A\| \ |\Qb_m(t)| + |\Vb_m(t)|  
\geq c_m/8 \right\}, T\right).
\end{equation}
\end{definition}
In the following, we assume that $m$ is sufficiently small so that the
initial condition satisfies $\|A\| \ |\Qb(0)| + |\Vb(0)|  
< c_m/8,$ where we recall that the initial condition is independent of $m.$

In Appendix~\ref{wellposed} we show that almost surely $\tau_m > 0$
and that the mechanical process is well-posed up to $\tau_m.$  In
particular we show that almost surely, before time $\tau_m$ there are
only finitely many collisions and none of these collisions involve
more than one bath atom at the same time. Intuitively, this means that
the fast collisions have no memory of the large particle's trajectory,
and furthermore it means that the rate of fast collisions is governed
only by the initial bath configuration~\eqref{bath_meas2}.  In
Appendix~\ref{norecollide}, we show that if a bath atom experiences a
fast collision with the mechanical process's large particle sometime
in the interval $[0, \tau_m),$ then this is the only collision that
the bath atom undergoes in $[0, \tau_m).$ From the coupling that we
construct in Appendix~\ref{me_to_ma} between the mechanical process
and the Markov approximation and a consideration of the expected size of the
Markov approximation (\eqref{gm} and~\eqref{Gmp1}), we can show that
$\lim_{m \rightarrow 0} \PP(\tau_m = T) = 1,$ so that for the sake of
our convergence proof, on our time interval of interest, $[0,T],$
fast collisions do not introduce memory terms. We thus define below a
rate of collisions that the large particle experiences whenever the
bath has configuration distributed according to~\eqref{bath_meas2}.

\begin{lemma}
Suppose that a large particle is at {\em{all}} times surrounded by a bath of
atoms
distributed according to~\eqref{bath_meas2}.  
Then the instantaneous probability of a collision between the particle
at position $\Qb$ and velocity $\Vb$ and a bath atom with velocity
in the ball $B(\vb; d \vb)$ on the surface element 
$R^{d-1} d\Omega$ centered around $\qb = \Qb - R \eb_{\rm n}$ in the time interval $[t,t+dt]$ is given by
\begin{equation}
\label{bg_measure1}
\begin{split}
r_m(\vb, \eb_{\rm n}, \Qb, \Vb) \ d\vb \ d \Omega \ dt 
= &\lambda_m R^{d-1} \ \max(v_{\rm n} - V_{\rm n}, 0)  
f_m(\Qb-R \eb_{\rm n},\vb) \, d\vb \,
d\Omega \, dt
\end{split}
\end{equation}
where $R$ is the radius of the large particle, $\lambda_m = m^{-1/2}
\lambda$ is the expected number of bath atoms per unit volume, $v_{\rm n}
= \vb \cdot \eb_{\rm n}$ is the normal velocity of the incoming bath
atom, and $V_{\rm n} = \Vb \cdot \eb_{\rm n}$ is the normal velocity
of the large particle.
\end{lemma}
\begin{remark}
The assumption that the large particle always sees the same
environment is too strong, and indeed does not hold for the large particle
in the mechanical process.  However, as we show in
Appendix~\ref{norecollide}, for bath atoms that undergo fast
collisions with the large particle, their pre-collision distribution
does satisfy~\eqref{bath_meas2}.  Thus, the rate~\eqref{bg_measure1}
is correct for fast collisions of the mechanical process. 
\end{remark}
\begin{proof}
In Figure~\ref{fig:fast} we sketch the differential volume element of
size $R^{d-1} d \Omega (v_{\rm n} - V_{\rm n}) \, dt$ with base on the
surface of the particle and height determined by the velocity of the
incoming bath atom.  The velocity of the bath atom and the large particle
are constant on the time interval $dt$, so this surface gives the
infinitesimal rate $r_m$ of fast collisions.
\end{proof}

\begin{figure}
\centerline{\input{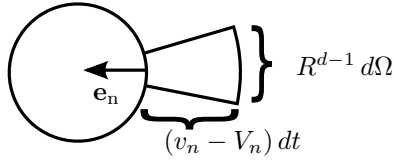}}
\caption{\label{fig:fast}The differential area carved out by fast
atoms that collide with the large particle in the time interval
$[t, t+dt]$, see~\protect\eqref{bg_measure1}. }
\end{figure}

\subsection{Rate of jumps and the Markov approximation}
We now use the rate~\eqref{bg_measure1} to define a rate for a Markov
approximation to the mechanical process, where we only count the
effect of fast collisions.  We change variables to be in terms of 
jumps in velocity for the large particle rather than velocities 
of bath atoms.  We first integrate out
tangential directions $\vb_t$ from $r_m,$
and define the marginal
\begin{equation}
\label{marge}
f^1(x) = \int_{\vb_t} f(x \eb_{\rm n} + \vb_t) \, d\vb_t,
\end{equation}
where we recall that $f$ was assumed to be rotationally invariant,
so that the left hand side does not depend on $\eb_{\rm n}.$ 
We define also the scaled marginal,
\begin{equation}
\label{fv1} 
f^1_m(x) =  m^{1/2} f^1(m^{1/2} x).
\end{equation} 
Rewriting the collision rule~\eqref{collision2}, we have
that the change in normal velocity of the large particle in the
$\eb_{\rm n}$-direction, $\Vh = V'_{\rm n} - V_{\rm n},$ due to a bath
atom with normal velocity $\vb_{\rm n}$ satisfies
\begin{equation}
\label{Vt2}
\Vh = \frac{2m}{M+m} (v_{\rm n} - V_{\rm n}).
\end{equation}
We integrate out the tangential
direction $\vb_t$ from the rate 
$r_m$, make the change of variables from $v_{\rm n}$ to $\Vh$
in~\eqref{Vt2}, and introduce a Heaviside function to restrict to
jumps caused by fast collisions.  The rate on jumps is
\begin{equation}
\begin{split} 
\label{bg_measure2}
\hat{r}_m(\Vh, \eb_{\rm n}, \Qb, \Vb)  
=& \lambda_m R^{d-1} \left(\frac{M+m}{2m}\right)^2 
H\left(\frac{M+m}{2m}\Vh + V_{\rm n}-(A\qb)_{\rm n} - c_m\right)\\
&  \times \max(\Vh,0) f^1_m\left(\frac{M+m}{2m}\Vh + V_{\rm n} - (A \qb)_{\rm n} \right),
\end{split} 
\end{equation}  
where  
\begin{equation*} 
\qb = \Qb - R \eb_{\rm n}
\end{equation*} 
represents the position on the large particle surface where the bath
atom collides.  Due to the Heaviside function, the minimum jump size
$\Vh$ with nonzero probability is
\begin{equation}
\label{vhm}
\Vhm = \frac{2 m}{M+m} \max\left(\left( c_m - V_{\rm n} + (A\qb)_{\rm
n} \right),0 \right),
\end{equation}
where we have kept implicit the dependence of $\Vhm$ on $\Qb$ and
$\Vb.$
We let 
\begin{equation}
\label{intensity}
\Lambda_m(\Qb,\Vb) = \int_{\SSph^{d-1}} \int_{\Vhm}^\infty 
\rh_m(\Vh, \eb_{\rm n}, \Qb, \Vb) \, d \Vh \, d \Omega
\end{equation} 
denote the jump intensity, and a short calculation shows that 
$\Lambda_m(\Qb, \Vb)$ is independent of $(\Qb, \Vb)$
whenever
\begin{equation*}
||A|| \ |\Qb| + |\Vb| \leq c_m/8
\end{equation*}
 holds, due to the Heaviside function in $\rh_m.$
\begin{definition}
\label{def:mark}
We define the Markov process $(\Qtb_m, \Vtb_m)$
starting from initial condition 
$(\Qtb_m(0), \Vtb_m(0))=(\Qb_m(0), \Vb_m(0))$ to be 
a Poisson jump process where the velocity experiences jumps
$\Vh \eb_{\rm n}$ with intensity 
$\Lambda_m(0,0)$ and with probability density function 
\begin{equation*}
\Lambda_m(0, 0)^{-1} \rh_m(\Vh, \eb_{\rm n}, \Qtb_m,\Vtb_m),
\end{equation*} 
and between jumps, $(\Qtb, \Vtb)$ evolve according
to~\eqref{particle_dynamics}.
\end{definition}

\section{Convergence of the Markov approximation to the SDE}
\label{ma_to_ou}

To show Theorem~\ref{thm:main}, we first show that the Markov
approximation $(\Qtb_m, \Vtb_m),$ with rate given
by~\eqref{bg_measure2}, converges in law to~\eqref{ou_flow}.  We write
the generator for the Markov approximation and show that it converges
as $m \rightarrow 0$ to the generator for~\eqref{ou_flow}.  These
calculations are carried out explicitly below, allowing us to get the
coefficients for the SDE in terms of the parameters of the heat bath
and the large particle.  The setup and calculations follow as
in~\cite{durr80}.  Because of the $\Qb$-dependence of the process, we
work with test functions $\psi(\Qb,\Vb)$ that are functions of the
position and velocity.  The inhomogeneity of the velocity field does
not change the coefficients $\gamma$ and $\sigma$ compared to the zero
background flow case~\eqref{lang_coeff}.
\begin{lemma}
\label{thm:ma_to_ou} 
Let $T>0$ be given.  The Markov process $(\Qtb_m, \Vtb_m)$ satisfies
\begin{equation*} 
(\Qtb_m, \Vtb_m)_\tint \lconv (\Qb, \Vb)_\tint 
\quad \text{ as } m\rightarrow 0
\end{equation*}
on $[0,T],$ where $(\Qb, \Vb)$ satisfies the nonequilibrium Langevin
dynamics~\eqref{ou_flow}.
\end{lemma}

For a stochastic process, we let $T_t: C_0(\RR^{2d};\RR) \rightarrow
C_0(\RR^{2d};\RR)$  denote the evolution semigroup, where
$C_0(\RR^{2d};\RR)$ denotes the set of continuous functions with zero
limit at infinity.  We define the infinitesimal generator $L$ by
\begin{equation} \label{gendef} L \psi(\Qb, \Vb) = \lim_{t\rightarrow
0} \frac{T_t \psi(\Qb, \Vb) - \psi(\Qb, \Vb)}{t} \end{equation} where
the domain of $L$ is the set of $C_0$ functions such that the above
limit exists in $C_0.$  In the following, we use $L_m$ to denote the
infinitesimal generator for the Markov approximation $(\Qtb_m,
\Vtb_m)$ and $L$ for the nonequilibrium Langevin
dynamics~\eqref{ou_flow}. 
We prove Lemma~\ref{thm:ma_to_ou} with the use of
Lemma~\ref{thm:gen_converge} below, which is a general lemma relating
convergence of generators to convergence in law of the process.  We
recall that $\D([0,T])$ represents the space of c\`adl\`ag functions
on the interval $[0,T].$ A linear subspace $K$ of the domain of $L$ is
called a core for $L$ if $L$ is the closure of $L|_K$~\cite{kurt75}.
One can show that $C^\infty_c,$ the set of compactly supported
$C^\infty$ functions, is a core for the nonequilibrium Langevin
dynamics~\eqref{ou_flow} using the ideas of~\cite{kurt75}.

The process $(\Qb, \Vb)$ is a Markov-$C_0$-process, which means its
evolution semigroup satisfies $T_t C_0
\subset C_0$ and $\lim_{t\rightarrow 0} ||T_t \psi - \psi|| =0$ for all 
$\psi \in C_0.$

\begin{lemma}
\label{thm:gen_converge}
Consider a family of Markov processes $Z_{m}$ on $\D([0,T])$ with
generators $L_{m}$.  Suppose that the Markov-$C_0$-process $Z$ has
generator $L$.  Let $K$ be a core for $L$ such that 
$\psi \in K \Longrightarrow \psi \in \D(L_m)$ for all sufficiently
small $m.$ Suppose that the initial distribution of $Z_{m}$
converges weakly to that of $Z$ and that
\begin{equation}
\label{gen_converge}
\forall \psi \in K, \qquad \lim_{m \rightarrow 0} 
   \sup_{z \in \RR^{2d}} | L_m \psi(z) - L \psi(z) | = 0 .
\end{equation} 
Then
\begin{equation*}
Z_{m} \lconv Z \quad m \rightarrow 0. \qed
\end{equation*} 
\end{lemma}
The above lemma is~\cite[Lemma 4.1]{durr80}, which can be found in
similar form in~\cite{kurt75} or~\cite{skor58}.

\subsection{Generator for the Markov process}
\label{sec:mark_gen}
We apply the generator $L_m$ for the Markov process to a
$C^{\infty}_c(\RR^{2d})$ test function, expand in
powers of $m,$ and show that we have, to leading order, a drift and
diffusion term.

Applying~\eqref{gendef} to the Markov process, we have for any
$\psi(\Qb,\Vb) \in C^\infty_c(\RR^{2d})$ that the generator satisfies
\begin{equation*}
\begin{split} 
L_m \psi(\Qb,&\Vb) = \Vb \cdot \nabla_\Qb \psi(\Qb,\Vb)
+ \int_{\SSph^{d-1}} \int_{\Vhm}^\infty 
              [\psi(\Qb, \Vb + \Vh \eb_{\rm n}) -\psi(\Qb, \Vb)] \ 
                 \hat{r}_m(\Vh, \eb_{\rm n}, \Qb, \Vb) \, d \Vh
d\Omega,
\end{split} 
\end{equation*}
where $\rh_m$ is defined in~\eqref{bg_measure2} and $\Vhm$, which we
note depends on $\eb_{\rm n},$ is defined in~\eqref{vhm}.  We recall,
from the scaling of bath atom velocities~\eqref{fvq}, that
$\EE_m(|\vb|)$ is $O(m^{-1/2}).$  From~\eqref{Vt2} this scaling
implies that $\EE_m(\Vh)$ is $O(m^{1/2})$.    To find the limiting
generator in the limit $m \rightarrow 0,$ we expand in powers of $\Vh$
and write 
\begin{equation*}
\begin{split} 
 \psi(\Qb, \Vb + \Vh \eb_{\rm n}) - \psi(\Qb, \Vb)
   =& \Vh \eb_{\rm n} \cdot \nabla_\Vb \psi(\Qb, \Vb) 
	+ \frac{1}{2} \ \Vh^2 (\eb_{\rm n} \otimes \eb_{\rm
n}) : \nabla_{\Vb}^2 \psi(\Qb, \Vb)  \\
   &+ \frac{1}{6} \Vh^3 (\eb_{\rm n} \otimes \eb_{\rm n} 
          \otimes \eb_{\rm n}) \cdot3\cdot \nabla_\Vb^3 
	      \psi(\Qb,\Vb + \Vh^* \eb_{\rm n}),
\end{split} 
\end{equation*} 
for some $\Vh^* \in (0, \Vh)$ and where $\cdot3\cdot$ denotes the third-order
contraction product, which we apply to the 3-tensors 
$(\eb_{\rm n} \otimes \eb_{\rm n} \otimes \eb_{\rm n})$ and $\nabla_\Vb^3.$
We then write
\begin{equation}
\label{Lm_expand}
\begin{split} 
L_m \psi (\Qb,\Vb) =& \Vb \cdot \nabla_\Qb \psi (\Qb,\Vb) 
	+ C_m \left( I_1 \cdot \nabla_\Vb \psi (\Qb,\Vb)  
+ \frac{1}{2} \  I_2 : \nabla_{\Vb}^2 \psi (\Qb,\Vb)
+ \mathcal{R}_m \right),
\end{split} 
\end{equation}
where, recalling that $\qb = \Qb - R \eb_{\rm n},$ we have
\begin{align}
\notag 
I_1 &= m^{-5/2} \int_{\SSph^{d-1}} \int_\Vhm^\infty
     \Vh^2 \eb_{\rm n} \ 
     f^1_m\left(\frac{M+m}{2m}\Vh + (\Vb - A \qb)_{\rm n}\right) d\Vh
     d\Omega, \\
\notag 
I_2 &= m^{-5/2} \int_{\SSph^{d-1}} \int_\Vhm^\infty
     \Vh^3 (\eb_{\rm n} \otimes \eb_{\rm n}) \ 
     f^1_m\left(\frac{M+m}{2m}\Vh + (\Vb - A \qb)_{\rm n}\right)  d\Vh
d\Omega,\\
\notag 
\mathcal{R}_m &= \frac{m^{-5/2}}{6} \int_{\SSph^{d-1}} \int_\Vhm^\infty
     \Vh^4 (\eb_{\rm n} \otimes \eb_{\rm n} \otimes \eb_{\rm n})
\cdot3\cdot \nabla_{\Vb}^3 \psi (\Qb,\Vb+\Vh^* \eb_{\rm n})  
f^1_m\left(\frac{M+m}{2m}\Vh + (\Vb - A \qb)_{\rm n}\right)  d\Vh
     d\Omega,
\end{align} 
with coefficient
\begin{align}
\label{cm}
C_m &= \lambda R^{d-1} \left(\frac{M+m}{2}\right)^2.
\end{align}
Note that in our notation, we have suppressed the dependence of $I_1,$
$I_2,$ and $\mathcal{R}_m$ on $(\Qb, \Vb)$.

\subsubsection{Remainder term}
\label{sec:mark_rem}
We begin with the remainder term $\mathcal{R}_m.$  Since the test
function $\psi$ belongs to $C^{\infty}_c,$ we can restrict to the case
where $\Qb$ and $\Vb$ are bounded, and we may assume $m$ is
sufficiently small so that $|\Vb| + \|A\| \, |\Qb| \leq c_m.$ This,
along with the finiteness of the moments~\eqref{phi_i_2}, allows us to
estimate the order of $\mathcal{R}_m.$  Letting 
$x= \left(\frac{M+m}{2m^{1/2}} \Vh + m^{1/2} (\Vb-A\qb)_{\rm n}\right),$
we compute
\begin{equation*}
\begin{split}
|\mathcal{R}_m|
&\leq C m^{-2} \int_{\SSph^{d-1}}  \int_\Vhm^\infty
\Vh^4 \ \left|\nabla^3_{\Vb} \psi(\Qb, \Vb + \Vh^* \eb_n)\right| \ 
f^1\left(\frac{M+m}{2m^{1/2}}\Vh + m^{1/2} (\Vb - A \qb)_{\rm n}\right) 
	 d\Vh d\Omega\\
&\leq \frac{C m^{1/2}}{(M+m)^5} 
     \int_{\SSph^{d-1}}  \int_{m^{1/2} c_m}^\infty (x - m^{1/2}  
     (\Vb- A\qb)_{\rm n})^4 \ \left|\nabla^3_{\Vb} \psi(\Qb, \Vb +
\Vh^* \eb_n)\right| \  f^1\left(x\right) \, dx \, d\Omega.
\end{split}
\end{equation*}
We note that since we have assumed that $m$ is sufficiently small,
the minimum value of $x$ is given by $\xm = m^{1/2} c_m.$
Recalling that $\psi$ is compactly supported and that 
$\Vh^* \in (0, \Vh) = \left(0, 
\left( \frac{2 m^{1/2}}{M+m} \right) \left( x - m^{1/2}
(\Vb-A\qb)_{\rm n}\right) \right),$
we can bound $\Qb, \Vb$ in the innermost integrand to find the
estimate
$(x - m^{1/2}  (\Vb- A\qb)_{\rm n})^4  \ \left|\nabla^3_{\Vb} \psi(\Qb, \Vb +
\Vh^* \eb_n)\right| \leq C (1+x)^4$ where $C$ may
depend on $\psi$ but not on $m.$  
We find an upper bound on $\mathcal{R}_m$ by extending the
interval of integration and using the boundedness of the
marginals~\eqref{phi_i_2},
\begin{equation}
\label{I3}
\begin{split}
| \mathcal{R}_m| &\leq 
C \frac{m^{1/2}}{(M+m)^5} \int_0^\infty (1+x)^4 f\left(x\right) dx  
	\leq C m^{1/2}.
\end{split}
\end{equation}

\subsubsection{Diffusion coefficient}
\label{sec:diffusion}
We next turn to $I_2$, and compute
\begin{equation*}
\begin{split}
I_2 &= m^{-5/2} \int_{\SSph^{d-1}}  \int_\Vhm^\infty 
	(\eb_{\rm n} \otimes \eb_{\rm n}) \Vh^3 \ 
	f^1_m\left(\frac{M+m}{2m}\Vh + (\Vb - A \qb)_{\rm n}\right) 
	 d\Vh d\Omega\\
&= m^{-2} \int_{\SSph^{d-1}}  \int_\Vhm^\infty 
	(\eb_{\rm n} \otimes \eb_{\rm n}) \Vh^3 
	f^1\left(\frac{M+m}{2m^{1/2}}\Vh + m^{1/2} 
	(\Vb - A \qb)_{\rm n}\right) d\Vh d\Omega.\\
\end{split}
\end{equation*}
Let $x= \left(\frac{M+m}{2m^{1/2}} \Vh + m^{1/2} (\Vb-A\qb)_{\rm n}\right).$  
We expand in powers of $m,$ using the finiteness of the moments~\eqref{phi_i_2} 
and the fact that $\psi$ is compactly supported, giving
\begin{equation*}
\begin{split} 
I_2 &= \frac{16}{(M+m)^4} \int_{\SSph^{d-1}} 
\int_{m^{1/2} c_m}^{\infty} 
	(\eb_{\rm n} \otimes \eb_{\rm n}) (x - m^{1/2}  (\Vb- A\qb)_{\rm n} )^3 \ 
	f^1(x) \, dx \, d\Omega \\
&= \frac{16}{M^4} \int_{\SSph^{d-1}}  \int_0^{\infty} 
	(\eb_{\rm n} \otimes \eb_{\rm n}) x^3 \ 
	f^1(x) \, dx \, d\Omega + O(m^{1/2}) \\
&= \frac{16}{M^4} \Phi_3 \int_{\SSph^{d-1}}  
	(\eb_{\rm n} \otimes \eb_{\rm n}) 
      \, d\Omega + O(m^{1/2}) \\
&= \frac{16 S_{d-1}}{ M^4 d} \Phi_3 I + O(m^{1/2}), \\
\end{split}
\end{equation*}
where $I$ denotes the $d \times d$ identity matrix and we recall that
$S_{d-1}$ denotes the surface area of the $(d-1)$-sphere $\SSph^{d-1}$.  The
second line uses the estimate $\int_0^{m^{1/2} c_m} x^3 f(x) \, dx
\leq  m^{3/2} c_m^3 \int_\RR f(x) \, dx = O(m^{1/2})$ 
along with a gathering of
higher order terms in $m.$ The third line uses the
definition~\eqref{phi_i_2} of $\Phi_3$, and the fourth line uses the
following identity
\begin{equation*}
 \int_{\SSph^{d-1}}  
	(\eb_{\rm n} \otimes \eb_{\rm n}) 
      \,  d\Omega = \frac{S_{d-1}}{d} I.
\end{equation*} 

Multiplying by $C_m$ of~\eqref{cm} and letting $m\rightarrow 0$ leads
to the isotropic diffusion coefficient
\begin{equation}
\label{D}
D I = \lim_{m \rightarrow 0} C_m I_2 
= \frac{4 \lambda R^{d-1} S_{d-1}}{ M^2 d} \Phi_3 I.
\end{equation}
In the case of $f(v) = Z^{-1} \exp\left(-\frac{\beta v^2}{2}\right)$
with
$Z= \frac{\sqrt{2 \pi}}{\sqrt{\beta}},$ 
we have $\Phi_3 = \frac{\sqrt{2}}{\beta^{3/2} \sqrt{\pi}}$ and
\begin{equation*}
D = \frac{4 \sqrt{2} \lambda R^{d-1} S_{d-1}}{ \beta^{3/2} \sqrt{\pi} M^2 d}. 
\end{equation*}

\subsubsection{Drift coefficient}
\label{sec:drift}
We now similarly expand $I_1$ and find that the lowest order term, which is $O(m^{-1/2}),$ cancels out leaving an $O(1)$ drift term.  Indeed, we have
\begin{equation}
\label{i1}
\begin{split} 
I_1  
&= \int_{\SSph^{d-1}} \eb_{\rm n} 
m^{-\frac{5}{2}} \int_\Vhm^\infty 	\Vh^2 \ 
	f^1_m\left(\frac{M+m}{2m}\Vh + V_{\rm n} - (A \qb)_{\rm n}\right) 
	 \,  d\Vh d\Omega. \\
\end{split} 
\end{equation}
As before, we let 
$x= \frac{M+m}{2m^{1/2}} \Vh + m^{1/2} (\Vb - A\qb )_{\rm n}.$  We obtain
\begin{equation*}
\begin{split} 
I_1=&\frac{ 8 m^{-1/2}}{(M+m)^3} 
  \int_{\SSph^{d-1}} \eb_{\rm n} \int_{m^{1/2} c_m}^\infty 
       (x - m^{1/2} (\Vb - A \qb)_{\rm n})^2 \ f^1 (x ) \,  dx \, d\Omega\\
=& \int_{\SSph^{d-1}} \eb_{\rm n} 
\left[ \frac{ 8 m^{-1/2}}{M^3} \int_{m^{1/2} c_m}^\infty 
       x^2 \ f^1 (x ) \, dx +\frac{ 16 }{M^3}  
  \int_{0}^\infty 
       (- x (\Vb-A\qb)_{\rm n}) \ f^1 (x ) \, dx \right] \, d\Omega +
        O(m^{3/10}), 
\end{split} 
\end{equation*}
where the error terms are dominated by
$\int_0^{m^{1/2} c_m} x f(x) \, dx
\leq  m^{1/2} c_m \int_\RR f(x) \, dx = O(m^{3/10}).$ The first term vanishes, 
\begin{equation*} 
\int_{\SSph^{d-1}} \eb_{\rm n}   
  \int_{m^{1/2} c_m}^\infty 
       x^2 \, f^1 (x )  \, dx \, d\Omega = 0,
\end{equation*} 
since $\int_{\SSph^{d-1}} \eb_{\rm n} \, d\Omega = 0$.
We recall 
$\qb = \Qb - R \eb_{\rm n},$ and let $\Wb = \Vb- A \Qb.$ 
We have the identities
\begin{equation*}
\begin{split}
&\int_{\SSph^{d-1}} \eb_{\rm n} \cdot \Wb \eb_{\rm n} \, d\Omega 
= \int_{\SSph^{d-1}} (\eb_{\rm n} \otimes \eb_{\rm n}) :
   \Wb  \, d\Omega = \frac{S_{d-1}}{d} \Wb,\\
&\int_{\SSph^{d-1}} \eb_{\rm n} (\eb_{\rm n} \cdot \eb_{\rm n}) \, d\Omega 
= \int_{\SSph^{d-1}} \eb_{\rm n} \, d\Omega = 0.\\
\end{split}
\end{equation*}
Combining the above, we have
\begin{equation*}
I_1 = - \frac{16 S_{d-1}}{M^3 d} (\Vb - A \Qb) \Phi_1 + O(m^{3/10}).
\end{equation*}
Multiplying by $C_m$ of~\eqref{cm} gives the drift coefficient
\begin{equation}
\label{gamma}
-\frac{\gamma}{M} (\Vb - A\Qb) = \lim_{m \rightarrow 0} C_m I_1 = - \frac{4
\lambda R^{d-1} S_{d-1}}{M d} \Phi_1 (\Vb-A \Qb).
\end{equation} 
In the case of $f(v) = Z^{-1} \exp\left(-\frac{\beta v^2}{2}\right)$
with
$Z= \frac{\sqrt{2 \pi}}{\sqrt{\beta}},$ we have
$\Phi_1=\frac{1}{\sqrt{2 \beta \pi}}$ and hence
\begin{equation*}
\gamma = \frac{2 \sqrt{2} \lambda R^{d-1} S_{d-1}}{\sqrt{\pi \beta} d}.
\end{equation*}

\subsection{Stochastic limit}

Combining the expansion of $L_m$ in~\eqref{Lm_expand}, with the
expressions for $I_1$ and $I_2$~\eqref{D}
and~\eqref{gamma} and with the bound on the remainder~\eqref{I3}, 
we have
\begin{equation*}
\begin{split}
L_m \psi(\Qb, \Vb) &= \Vb \cdot \nabla_Q \psi(\Qb, \Vb) - M^{-1} \gamma (\Vb- A \Qb) 
\cdot \nabla_V \psi(\Qb, \Vb)+ \frac{1}{2} \  D \Delta_{V} \psi(\Qb, \Vb) 
+ O(m^{3/10}).
\end{split}
\end{equation*}
We thus have a generator $L_m$ that in the limit $m \rightarrow 0$
converges in the sense of~\eqref{gen_converge} to the generator of the
nonequilibrium Langevin dynamics~\eqref{ou_flow}, where $\gamma$ is
given in~\eqref{gamma} and $\sigma = M \sqrt{D}$ is defined
using~\eqref{D}, in agreement with~\eqref{neld_coeff}.  Thus, we use
Lemma~\ref{thm:gen_converge} to conclude the convergence of the Markov
process as stated in Lemma~\ref{thm:ma_to_ou}.

\section{Fast collisions cannot lead to recollisions}
\label{norecollide}

We show in this section that in the mechanical process until the
stopping time~\eqref{stop_time_mech},  any bath atom that undergoes a
fast collision (recall Definition~\ref{fast}) cannot recollide with
the large particle and that the bath atom cannot have previously
collided with the large particle.  This shows the fast collisions
experienced by the large particle have rate $r_m$
in~\eqref{bg_measure1}.  This observation is necessary when coupling
the Markov approximation to the mechanical process, which we do in
Appendix~\ref{me_to_ma}.

\begin{lemma}
\label{thm:norecollide}
For a given trajectory $(\Qb_m, \Vb_m),$ suppose that the large
particle experiences a fast collision with a bath atom at time $t_1
\in [0,\tau_m).$ Then there are no other collisions between this atom
and the large particle in the time interval $[0, \tau_m),$ where
$\tau_m$ denotes the stopping time~\eqref{stop_time_mech}.  
\end{lemma}

\begin{proof}
We recall that the relative velocity of the bath atom is $\vcen = \vb
- A \qb.$   By the choice of bath dynamics~\eqref{nonHam2}, the
relative velocity only changes by colliding with the large particle
whereas the velocity $\vb$ changes with time according
to~\eqref{nonHam1}.  We consider a fast collision, 
\begin{equation*}
\vce_{\rm n} > c_m = m^{-1/5}.
\end{equation*}
We write the collision rule~\eqref{collision2} in terms of relative
velocity before making use of the bound on particle position and velocity
in~\eqref{stop_time_mech} (note that since the collision occurs on the
particle's surface, $\qb(t_1)$ can be bounded in terms of $\Qb(t_1)$).
We then have the following bound on the relative normal velocity of
the bath atom after the collision
\begin{equation*}
\begin{split}
\vce'_{\rm n} &= -\frac{M-m}{M+m} \vce_{\rm n} 
   + \frac{2M}{M+m} V_{\rm n} - \frac{2M}{M+m} (A \qb(t_1))_{\rm n} \\
&\leq - \left( \frac{M-m}{M+m} \right) c_m +  \frac{2M}{M+m} \left( \frac{c_m}{8} + R \|A\| \right) \\
&\leq -\frac{2}{3} c_m.
\end{split}
\end{equation*} 
The last line holds for $m$ sufficiently small.  This shows that after colliding, the 
bath atom moves away from the particle with a large velocity.

We let $\eb_{\rm n,1}$ denote the normal vector for the fast collision
at time $t_1$ and look at the future velocity of the bath atom,
$\vb(t) = A \qb(t) + \vcen'(t)$ for $t \in [t_1, \tau_m).$  From the
above computation we see that the relative velocity of the bath atom
is pointed away from the large particle in the $\eb_{\rm
n,1}$-direction.  Due to the bound on the position of the large
particle in~\eqref{stop_time_mech}, any recollision with the bath atom
must occur in the region $||A|| |\Qb| \leq c_m/8;$ however, whenever
the bath atom is in this region, its velocity in the $\eb_{\rm
n,1}$-direction satisfies $\vb(t) \cdot \eb_{\rm n,1} \leq (c_m/8 - 2
c_m/3)  \leq -\frac{c_m}{2}.$  The large particle's velocity is
bounded by $c_m/8,$ and is thus 
too
low to overtake the bath atom.  Therefore, there cannot be a
recollision before time $\tau_m$.  Likewise, before the collision, the
velocity in the $\eb_{\rm n,1}$ direction satisfied $\vb(t) \cdot
\eb_{\rm n,1} \geq \frac{3 c_m}{4},$ which is faster than the large
particle's speed which is bounded by $c_m/8$, so it is impossible that
there were previous collisions before time $t=0.$
\end{proof}

\section{Coupling the mechanical process to the Markov approximation}
\label{me_to_ma}

In this section, we couple the Markov process $(\Qtb_m, \Vtb_m)$ with
the rate~\eqref{bg_measure2} to the mechanical process $(\Qb_m,
\Vb_m),$ defined in Section~\ref{mechdef}.  That is, for each
realization of the mechanical process, we associate a set of
realizations of the Markov process that are, with high probability,
close in the $L^{\infty}$-norm.  For the coupling, we prove below that 
\begin{equation*}
(\Qtb_m, \Vtb_m)_\tint-(\Qb_m, \Vb_m)_\tint \pconv 0 
	 \textrm{ as } m \rightarrow 0.
\end{equation*}  
As in~\cite{durr80}, we define a stopping time  when the processes
first differ by $\eps > 0$  and bound the total effect of the velocity
jumps not shared between the two processes to show that in the limit
$m\rightarrow 0$ the stopping time is greater than or equal to $T$ with 
probability 1.  

As described below, for the majority of the velocity jumps in the
mechanical process caused by fast collisions, we subject the Markov
process to the same velocity jump.  The construction of the Markov
process differs slightly from that in~\cite{durr80}.  Here, we couple
jumps in the velocity rather than collisions with bath atoms, which
makes the argument simpler.  This simplification is made possible by
the additional assumption that $f(x)$ is decreasing, which we have
assumed in our case in order to properly handle the fact that the bath
atom velocity distribution depends on position.

\subsection{Coupling and convergence}
We construct the coupling up to the stopping time $\tau_m$ for the 
mechanical process~\eqref{stop_time_mech}, and extend the
definition of the Markov process up to time $T,$ if necessary.
Let $\ime = \{t_1, t_2, \dots \}$ denote the set of
all times up to $\tau_m$ at which 
the large particle in the mechanical process experiences a collision.
This set is shown to be almost surely finite for any $m$ in
Appendix~\ref{wellposed}.  
We define $\vb_{{\rm n}, i} = \vb_{{\rm n}}(t_i),$ 
$\vcen_{{\rm n}, i} = \vcen_{{\rm n}}(t_i),$ 
$\Vb_{{\rm n}, i} = \Vb_{{\rm n}}(t_i),$ 
$\eb_{{\rm n}, i} = \eb_{{\rm n}}(t_i),$ 
etc.
We let $\is =
\{t_i \in \ime : \ |\vcen_{{\rm n},i}| < c_m\}$ denote the set of jump times
due to slow
collisions and $\ifa = \ime \setminus \is$ denote the set of jump times due to fast
collisions.  These sets of jump times are random variables since the initial
condition is random.

For a given trajectory of the mechanical process $(\Qb_m, \Vb_m),$ we
define $(\Qtb_m(0), \Vtb_m(0)) = (\Qb_m(0), \Vb_m(0)),$ 
and for most of the fast collisions of the mechanical process we
subject the Markov process to the same jump in velocity that the
mechanical process undergoes.  We selectively accept some subset of the
fast collisions in the time interval $[0, \tau_m)$
and add additional jumps in the velocity to ensure that the jumps  of
the Markov process have the rate $\rh_m(\Vh, \eb_{\rm n}, \Qtb_m, \Vtb_m)$
defined in~\eqref{bg_measure2}.  (Recall that the Markov process has
no slow collisions.)  More precisely, for every $t_i \in
\ifa,$ we choose to apply a jump with velocity change 
$\Vh_i \eb_{{\rm n}, i}$  to the Markov process with probability 
\begin{equation}
\label{rate_rem}
p_{\rm keep}(\Vh_i,\eb_{{\rm n}, i},\Qtb_m, \Vtb_m, \Qb_m, \Vb_m) = \min \left( 
\frac{\rh_m(\Vh_i, \eb_{{\rm n}, i}, \Qtb_m, \Vtb_m)}{\rh_m(\Vh_i, \eb_{{\rm n}, i}, \Qb_m, \Vb_m)}, 1 \right).
\end{equation}
For $t \in [0, \tau_m),$ we add additional fast collisions 
to the Markov process with the Poisson rate
\begin{equation}
\label{rate_add}
\begin{split}
{r}_{\rm add}(&\Vh,\eb_{\rm n},\Qtb_m, \Vtb_m, \Qb_m, \Vb_m)
= \max \left( \rh_m(\Vh, \eb_{\rm n}, \Qtb_m, \Vtb_m) 
              - \rh_m(\Vh, \eb_{\rm n}, \Qb_m, \Vb_m), 0 \right).
\end{split}
\end{equation} 
After accepting collisions with probability~\eqref{rate_rem}, and
adding new collisions with rate~\eqref{rate_add}, a short
calculation shows that the rate function for the Markov process in the 
time interval $[0,\tau_m)$  is
$\rh_m(\Vh, \eb_{\rm n}, \Qtb_m, \Vtb_m).$  
If $\tau_m < T,$ we extend the Markov process
to $[0,T]$ by adding additional jumps with rate 
$\rh_m(\Vh, \eb_{\rm n}, \Qtb_m, \Vtb_m).$

We let $\ima$ denote the set of all jump times of the Markov
process, $\iex \subset \ifa$ denote the set of jump times of the mechanical
process that
were not chosen for the Markov process, and $\iext \subset \ima$ the
set of additional times at which jumps were added to the Markov process (in the
whole interval $[0,T]$).  From
the construction above, we have that 
$$\ima = (\ifa \setminus \iex) \cup \iext.$$  
We note that the sets of jump times $\ime, \ima, \is, \iex,$ and
$\iext$ are random variables.

Since the chosen realization of the mechanical process may 
not be defined on the full time interval of interest $[0,T],$
we make the convention that 
\begin{equation*}
\sup_{s \in [0,t]} 
|\Vtb_m(s) - \Vb_m(s)| = \infty 
 \quad         \text{ if } \tau_m < t.
\end{equation*}
In particular, part of the proof of convergence will be the fact 
that $\lim_{m \rightarrow 0} \PP( \tau_m = T) = 1.$
We then show the following convergence in probability result.
\begin{lemma}
\label{thm:strong}  For all $T>0,$ 
\begin{equation*}
(\Qtb_m, \Vtb_m)_\tint-(\Qb_m, \Vb_m)_\tint \pconv 0 \textrm{ as } m \rightarrow 0,
\end{equation*}
where the convergence in probability is with respect to the
$L^{\infty}([0,T])$-norm.
Equivalently, for all $T > 0$ and any $\eps > 0,$
\begin{equation*}
\lim_{m \rightarrow 0} \PP\left(  
\sup_{t \in [0,T]} |\Qtb_m(t) - \Qb_m(t)| + |\Vtb_m(t) - \Vb_m(t)| \geq \eps \right)  = 0.
\end{equation*} 
\end{lemma}
To show Lemma~\ref{thm:strong}, we prove the following lemma which
says that the interval of convergence can be extended by a finite time
step.
\begin{lemma}
\label{thm:strong_extend}
If $t_0 \geq 0$ is such that 
\begin{equation}
\label{pconv_hyp}
\lim_{m \rightarrow 0} \PP \left( 
 \sup_{t \in [0,t_0]} \left|\Qtb_m(t) - \Qb_m(t)\right| + \left|\Vtb_m(t) - \Vb_m(t) \right|  \geq \eps \right) = 0
\end{equation}
for all $\eps >0$, we then have 
\begin{equation} 
\label{pconv_conc}
\lim_{m \rightarrow 0} \PP \left(
\sup_{t \in [0,t_0 + z]} \left|\Qtb_m(t) - \Qb_m(t)\right| + \left|\Vtb_m(t) - \Vb_m(t)\right|  \geq \eps \right) = 0
\end{equation} 
for all $\eps > 0$ where 
$z = \min\left(M (192 \lambda R^{d-1} S_{d-1} \Phi_1)^{-1},
\frac{1}{2 (1+\|A\|)}\right).$
\end{lemma}
Since the initial conditions for the mechanical process and Markov
process are the same, the hypothesis of Lemma~\ref{thm:strong_extend}
holds for $t_0=0$.  Thus, Lemma~\ref{thm:strong} follows immediately
from Lemma~\ref{thm:strong_extend}, and we are left with proving
Lemma~\ref{thm:strong_extend}.  The computation justifying the
particular choice of $z$ is performed in Appendix~\ref{sec:overlapping}.

\subsection{Error decomposition}
We prove Lemma~\ref{thm:strong_extend} by splitting the error in the
velocity into two contributions and bounding them individually.
We fix $\eps >0$ and $t_0 >0,$ and in the following we denote by 
$\ime' = \ime \cap [t_0,t_0+z],$ $\ima'=\ima \cap [t_0,t_0+z],$
and likewise for $\is',$
$\iex',$ and $\iext'.$ 
For times $t \in [t_0,\tau_m),$ the error in velocity 
between the Markov and mechanical processes is
\begin{equation*}
\begin{split}
\left| \Vb_m(t) - \Vtb_m(t) \right| 
  =& \left| \Vb_m(t_0) 
  + \sum_{\substack{t_i \in \ime'}} \Vh_i \eb_{{\rm n}, i} \ 
        \one_{[t_0, t]}(t_i)  
  - \Vtb_m(t_0) 
  - \sum_{\substack{t_i \in \ima'}} \Vh_i \eb_{{\rm n}, i} \
        \one_{[t_0, t]}(t_i) \right|,
\end{split}
\end{equation*} 
where we recall that 
$t_i$ is the time that the jump $\Vh_i \eb_{{\rm n}, i}$ occurs and 
where $\one_S$ denotes the characteristic function of the set 
$S$.

We split the error into two sources.  The first source of error is the
change in velocity due to slow collisions of the mechanical process,
since the Markov process does not include any slow collisions.
We denote this contribution in the interval $[t_0, t]$ by
\begin{equation}
\label{wslow}
W_{\rm slow}(t) = \left| \sum_{\substack{t_i \in \is'}}  \Vh_i 
\eb_{{\rm n}, i} \ \one_{[t_0, t]}(t_i) \right|.
\end{equation}
The second source of error is the change in velocity due to the 
jumps that we added and removed when coupling the Markov process to
the mechanical process.  We define
\begin{equation}
\label{wex}
W_{\rm ex}(t) = \sum_{t_i \in \iex'}  \Vh_i \ \one_{[t_0, t]}(t_i)
	  + \sum_{t_i \in \iext'} \Vh_i \ \one_{[t_0, t]}(t_i).
\end{equation} 
By definition, $W_{\rm ex}(t) \geq 0$ for all $t$ (recall that $\Vh \geq
0$).  
The error terms $W_{\rm slow}$ and $W_{\rm ex}$
are random variables.

We decompose the error in the velocity as 
\begin{equation}
\label{err_decomp}
\begin{split} 
| \Vb_m(t) - \Vtb_m(t) |  
=& \left| \Vb_m(t_0) 
 + \sum_{t_i \in \ime'} \Vh_i  \eb_{{\rm n}, i} \one_{[t_0, t]}(t_i) 
- \Vtb_m(t_0) - \sum_{t_i \in \ima'} \Vh_i  \eb_{{\rm n}, i} 
                          \one_{[t_0, t]}(t_i)\right|\\[6pt]
\leq& \left| \Vb_m(t_0) - \Vtb_m(t_0) \right|  + \left| \sum_{t_i \in \is'} 
  \Vh_i  \eb_{{\rm n}, i} \one_{[t_0, t]}(t_i) \right|  \\
&+ \sum_{t_i \in \iext'} \left| \Vh_i  \eb_{{\rm n}, i} 
                          \one_{[t_0, t]}(t_i) \right|
 +  \sum_{t_i \in \iex'} \left|\Vh_i  \eb_{{\rm n}, i} 
                          \one_{[t_0, t]}(t_i) \right|  \\
=& \left| \Vb_m(t_0) - \Vtb_m(t_0) \right| + W_{\rm slow}(t) + W_{\rm ex}(t).
\end{split} 
\end{equation} 

We fix $\eps > 0$ and, as in~\cite{durr80}, define the stopping time
\begin{equation}
\label{vel_close}
\tau^*_m= \min\left(\inf_{t \in [t_0,\tau_m)} 
\left\{t : \left|\Vtb_m(t) - \Vb_m(t) \right| \geq \eps/2 \right\},
               t_0+z\right).
\end{equation}
The use of a stopping time gives us a bound on how close the 
mechanical and Markov processes are, which allows us to then 
bound the difference in jumps each process experiences.
Provided that we have $\left| \Qtb_m(t_0) - \Qb_m(t_0)\right| \leq \eps/4,$ 
which we can assume from the hypothesis in
Lemma~\ref{thm:strong_extend}, 
we have that
\begin{equation}
\label{pos_close}
\begin{split}
&\sup_{t \in [t_0, \tau^*_m)} \left|\Qtb_m(t) - \Qb_m(t) \right| 
 \leq \left| \Qtb_m(t_0) - \Qb_m(t_0)\right| 
   + z \sup_{t \in [t_0, \tau^*_m)} \left|\Vtb_m(t) - \Vb_m(t) \right|
\leq \eps/2,
\end{split}
\end{equation}
where we recall $z \leq \frac{1}{2 (1+\|A\|)}.$
Thus, we have the relations
\begin{align}
&\sup_{t \in [t_0, \tau^*_m)} \left| \Qtb_m(t) - \Qb_m(t)\right| +
\left|\Vtb_m(t) - \Vb_m(t) \right|  \leq \eps, \label{l2norm_small} \\
&\sup_{t \in [t_0, \tau^*_m)} \|A\| \, \left| \Qtb_m(t) - \Qb_m(t)\right| +
\left|\Vtb_m(t) - \Vb_m(t) \right|  \leq \eps. \label{anorm_small}
\end{align} 
From the above arguments, 
we may show Lemma~\ref{thm:strong_extend} by showing that
\begin{equation}
\label{tstar_con}
 \lim_{m \rightarrow 0} \PP( \{\tau^*_m < t_0 + z\}) = 0.
\end{equation}

It is difficult to estimate {\em {a priori}} the probability
$\PP(\tau_m < t_0 + z),$ which appears indirectly in the estimate
of $\PP(\tau^*_m < t_0 + z)$ through the definition of 
$\tau^*_m$~\eqref{vel_close}.  To aid in the estimate we
define the set of trajectories where the 
Markov process remains small compared 
to~$c_m$ is 
\begin{equation}
\label{gm}
G_m = \left\{\sup_{0 \leq t \leq T} \left( \|A\| \ 
\left|\Qtb_m(t)\right| + \left|\Vtb_m(t)\right| \right) \leq c_m/16 \right\}.
\end{equation} 
From the convergence of the Markov process in
Lemma~\ref{thm:ma_to_ou}, we have that 
\begin{equation}
\label{Gmp1}
\lim_{m\rightarrow 0} \PP(G_m) = 1. 
\end{equation}
By the triangle inequality, we see that for trajectories
belonging to $G_m,$ $\tau_m^* < t_0 + z \Rightarrow 
\left|\Vtb_m(\tau^*) - \Vb_m(\tau^*) \right| \geq \eps/2.$
Thus, in order to have~\eqref{tstar_con}, it is sufficient to show that 
\begin{equation}
\label{gm_cap}
\lim_{m \rightarrow 0} \PP( \{\tau^*_m < t_0 + z\} \cap G_m) = 0.
\end{equation}

To show~\eqref{gm_cap}, we use~\eqref{err_decomp} to break apart the terms:
\begin{equation}
\label{set_decomp}
\begin{split}
\{\tau^*_m < t_0 + z\} \cap G_m \subset 
     & \left\{\{W_{\rm slow}(\tau^*_m) \geq \eps/6\}\cap G_m \right\}
     \cup \left\{\{W_{\rm ex}(\tau^*_m) \geq \eps/6\} \cap G_m \right\} \\
     &\cup \left\{\left|\Vtb_m(t_0) - \Vb_m(t_0)\right| \geq \eps/6 \right\}.
\end{split}
\end{equation}  
In terms of probability, we have
\begin{equation}
\label{prob_decomp}
\begin{split}
\PP(\{\tau^*_m < t_0 + z\} \cap G_m) 
\leq& \, \PP(\{W_{\rm slow}(\tau^*_m) \geq \eps/6\} \cap G_m)
        +\PP(\{W_{\rm ex}(\tau^*_m) \geq \eps/6\} \cap G_m)\\
&+\PP\left(\left|\Vtb_m(t_0) - \Vb_m(t_0)\right| \geq \eps/6 \right).
\end{split}
\end{equation}  
By the hypothesis of Lemma~\ref{thm:strong_extend}, 
$$
\lim_{m \rightarrow 0} \PP\left(\left\{\left|\Vtb_m(t_0) - \Vb_m(t_0)\right| \geq \eps/6 \right\}\right) = 0.
$$
We prove convergence of the other two terms in the following sections.

\subsection{Different collision rates}
\label{sec:collision_rate}

Since we restrict our attention to times $t < \tau_m,$ Appendix~\ref{norecollide}
shows that the rate of jumps due to fast collisions 
in the mechanical process is given by~\eqref{bg_measure2}.  We now estimate  
\begin{equation*}
W_{\rm ex}(\tau^*_m) = \sum_{t_i \in \iex'} \Vh_i  \ \one_{[t_0,
\tau^*_m]}(t_i) 
+ \sum_{t_i \in \iext'} \Vh_i \ \one_{[t_0, \tau^*_m]}(t_i),
\end{equation*}
which bounds the total effect of the jumps that were added and removed
in the coupling process using the expressions in~\eqref{rate_rem}
and~\eqref{rate_add}.  
These jumps have the rate function
\begin{equation*}
\begin{split}
r_{\rm ex}\left(\Vh, \eb_{\rm n}, \Qtb_m, \Vtb_m,\Qb_m, \Vb_m\right)
= \left| \rh_m\left(\Vh, \eb_{\rm n}, \Qtb_m, \Vtb_m\right) 
          - \rh_m\left(\Vh, \eb_{\rm n}, \Qb_m, \Vb_m\right)\right|.
\end{split}
\end{equation*}
The goal of this subsection is to show the following.
\begin{lemma}
\label{thm:wex}
Fix $\eps > 0$, and let $\tau^*_m$ be given by~\eqref{vel_close}. Let
$G_m$ be given by~\eqref{gm},
and let $W_{\rm ex}$ be given by~\eqref{wex}.  We have that
\begin{equation*} 
\lim_{m \rightarrow 0}
\PP(\{W_{\rm ex}(\tau^*_m) \geq \eps/6\} \cap G_m) = 0.
\end{equation*}
\end{lemma}

\begin{proof}
In the following, we bound the error $W_{\rm ex}$ in terms of a pair
of compound Poisson processes, whose definition we recall below.
The error  $W_{\rm ex}$ is
not a compound Poisson process itself, due to the dependence of
$r_{\rm ex}$ on $\Qtb_m, \Vtb_m, \Qb_m, \Vb_m,$ and $\eb_{\rm n},$ 
but we build a compound Poisson process from
realizations of $W_{\rm ex},$  by increasing the size of
certain jumps as well as adding additional jumps as detailed below.

\subsubsection{Compound Poisson Processes}

We begin by recalling a few properties of compound Poisson processes
before proving our bound on $W_{\rm ex}.$  (Definitions may be found
for instance in~\cite{cont04}.)
\begin{definition}
A {\em {compound Poisson process}} $J$ is a stochastic process defined
in terms of its {\em{ rate function}} $r(\check{V})$ by
\begin{equation*}
J(t) = \sum_{i=1}^{N(t)} \check{V}_i
\end{equation*}
where $N(t)$ is a Poisson process with rate
$\Lambda := \int_{\RR} r(\check{V}) \, d \check{V}$ that is
independent from the jumps $\{\check{V}_i\}_{i=1}^{N}$ which are 
independent, identically distributed 
random variables with probability density function
$\Lambda^{-1} r(\check{V}).$
\end{definition}

The expected value of a compound Poisson process satisfies
\begin{equation}
\label{comp_ev}
\begin{split}
\EE(J(t)) &= \EE( N(t)) \EE( \check{V} )
= t \int_\RR \check{V} r(\check{V}) \, d \check{V} \, dt. 
\end{split}
\end{equation}
Using the independence of $\{\check{V}_i\}_{1 \leq i \leq N_{\rm ex}}$ 
and $N_{\rm ex},$ we can simplify the variance
\begin{equation}
\label{comp_var}
\begin{split}
\EE( (J(t) - \EE(J(t)))^2) 
&= \EE\left( \left(\sum_{i=1}^{N(t)} \check{V}_i - \EE(J(t))\right)^2\right) \\
&= \EE\left( \left(\sum_{i=1}^{N(t)} \check{V}_i 
         - \EE(N(t))\EE(\check{V})\right)^2\right) \\
&= \EE\left( \left(\sum_{i=1}^{N(t)} 
      (\check{V}_i - \EE(\check{V})) + [N(t) - \EE(N(t))]
       \EE(\check{V})\right)^2\right) \\
&= \EE\left( \sum_{i=1}^{N(t)} (\check{V}_i - \EE(\check{V}))^2\right) 
	+ \EE\left( (N(t) - \EE(N(t))^2\right)\EE(\check{V})^2 \\
&= \EE(N(t)) \EE([\check{V} - \EE(\check{V})]^2) 
	+ \EE(N(t)) \EE(\check{V})^2 \\
&= \EE(N(t)) \EE(\check{V}^2)\\
&= t \int_\RR \check{V}^2 r(\check{V}) \, d \check{V}. 
\end{split}
\end{equation}
To go from the third to the fourth line, we have used
independence of $N(t)$ and $\check{V_i}$, 
and to go from the fourth to the fifth we have
used the independence of the jumps $\{\check{V}_i\}_{i=1}^{N(t)}$
and the following property
of the variance of exponential distributions:
\begin{equation*}
\EE\left( (N(t) - \EE(N(t))^2\right)
=\EE(N(t)). 
\end{equation*}

\subsubsection{Building the auxiliary process}
Since there are Heaviside functions within the definition of $\rh_m,$
we consider three regions: one where both rate functions are nonzero,
a second where one rate function is identically zero, and a third
where they are both identically zero.  For $t \in [t_0, \tau^*_m),$ we
define
\begin{equation*}
\begin{split}
v_1 &= \frac{2 m}{M+m} \min\left(\left( c_m - \Vt_{\rm n} 
      + (A\qtb)_{\rm n} \right),
      \left( c_m - V_{\rm n} + (A\qb)_{\rm n} \right) \right) \\
v_2 &= \frac{2 m}{M+m} \max\left(\left( c_m - \Vt_{\rm n} 
      + (A\qtb)_{\rm n} \right),
      \left( c_m - V_{\rm n} + (A\qb)_{\rm n} \right) \right) 
\end{split}
\end{equation*}
where we recall that $\qb = \Qb_m - R \eb_{\rm n}$ 
and $\qtb = \Qtb_m - R \eb_{\rm n}$ 
are the locations of the collision on the surface of the respective
atoms.  Note that 
$\frac{2m}{M+m} \frac{7 c_m}{8} \leq v_1 
            \leq v_2 \leq \frac{2m}{M+m} \frac{9 c_m}{8}$
from~\eqref{gm} and~\eqref{stop_time_mech}, 
and $v_2 - v_1 \leq \frac{2m}{M+m} \eps$ from~\eqref{anorm_small}.  

In the range $v_1 \leq \Vh \leq v_2,$ we can directly bound 
$r_{\rm ex}$ using the monotonicity of $f$ 
\begin{equation*}
\begin{split}
r_{\rm ex}(\Vh, \eb_{\rm n}, \Qtb_m, \Vtb_m,\Qb_m, \Vb_m) 
&= C_m m^{-5/2} \Vh 
    f^1_m\left(\frac{M+m}{2m} (\Vh-v_1) + c_m \right) \\
&\leq C_m m^{-5/2} \Vh 
    f^1_m\left(\frac{M+m}{2m} \Vh - \frac{c_m}{8} \right) 
\end{split}
\end{equation*}
where 
\begin{equation}
\label{Cm_couple}
C_m = \lambda R^{d-1} \left(\frac{M+m}{2}\right)^2.
\end{equation}
We integrate out $\eb_{\rm n},$ giving the factor $S_{d-1},$
 and define the rate 
\begin{equation*}
\begin{split}
r_1(\Vh) &= C_m S_{d-1} m^{-5/2} \Vh 
    f^1_m\left(\frac{M+m}{2m} \Vh - \frac{c_m}{8} \right).
\end{split}
\end{equation*}
To construct the Poisson process, we add additional jumps in $\Vh$ with the (positive) rate 
$$r_1(\Vh) - \int_{\SSph^{d-1}} 
        r_{\rm ex}(\Vh, \eb_{\rm n}, \Qtb_m, \Vtb_m,\Qb_m, \Vb_m) \ d \Omega.$$

For $v_2 \leq \Vh < \infty,$ we have the following estimate:
\begin{equation*}
\begin{split}
r_{\rm ex}(&\Vh, \eb_{\rm n}, \Qtb_m, \Vtb_m,\Qb_m, \Vb_m)\\ 
&= C_m m^{-5/2} \Vh   \left| 
    f^1_m\left(\frac{M+m}{2m}\Vh + \Vt_{\rm n} - (A \qtb)_{\rm n} \right) \right. 
\left. - f^1_m\left(\frac{M+m}{2m}\Vh + V_{\rm n} - (A \qb)_{\rm n} \right) \right|  \\
&= C_m m^{-5/2} \Vh   \left[ 
    f^1_m\left(\frac{M+m}{2m} (\Vh - v_2) + c_m \right) \right. 
     \left. - f^1_m\left(\frac{M+m}{2m} (\Vh -v_1) + c_m\right) \right] \\
&\leq C_m m^{-5/2} \Vh   \left[ 
    f^1_m\left(\frac{M+m}{2m} (\Vh - v_2) + c_m \right) \right. 
     \left. - f^1_m\left(\frac{M+m}{2m} (\Vh -v_2) + c_m + \eps \right) \right]
\end{split}
\end{equation*}
where the inequality uses the monotonicity of $f^1.$ We cannot
directly eliminate $v_2$ from the argument by relying on monotonicity
arguments since we
have a difference of $f,$ so we shift by replacing jumps generated by
the process corresponding to $r_{\rm ex}$ by larger jumps.   For any
jump of size $v_2 \leq \Vh < \infty,$ we replace it with a new, larger jump of
size $\check{V} = \Vh + \frac{2m}{M+m} \frac{9 c_m}{8} - v_2.$ We also
increase the rate of jumps by a factor of $\frac{\check{V}}{\Vh}$ and
integrate out $\eb_{\rm n},$ giving the rate
\begin{equation*}
\begin{split}
r_2(\check{V}) &= C_m S_{d-1} m^{-5/2} \check{V}   \left[ 
    f^1_m\left(\frac{M+m}{2m} \check{V} - \frac{c_m}{8} \right) 
  - f^1_m\left(\frac{M+m}{2m} \check{V} - \frac{c_m}{8} + \eps\right) \right].
\end{split}
\end{equation*}

The limits of the regions above, $v_1$ and $v_2,$ also depend on 
$(\Qb_m,\Vb_m)$ and $(\Qtb_m, \Vtb_m)$.  To remove this dependence
we extend the interval of definition for the two rate functions
defined above to 
$\frac{2m}{M+m} \frac{7 c_m}{8} \leq \Vh 
               \leq \frac{2m}{M+m} \frac{9 c_m}{8}$  for $r_1$ 
and to $\frac{2m}{M+m} \frac{7 c_m}{8} \leq \check{V} < \infty$
for $r_2.$  This can be done by adding additional jumps to the 
auxiliary Poisson processes above. 
All told, we have build realizations of two compound Poisson processes
\begin{equation*}
J_{m,k}= \sum_{i=1}^{N_{{\rm ex}, k}} \check{V}_i  
\ \one_{[t_0, t_0+z]}(t_i) 
\qquad \text{ for } k=1,2,
\end{equation*}
with the respective rate functions
\begin{equation*}
\begin{split}
\check{r}_{\rm ex,1}(\check{V}) &= r_1(\check{V}) \ \one_{[\frac{7 c_m}{8},\frac{9 c_m}{8}]} 
\left( \frac{M+m}{2m} \check{V}\right), \\
\check{r}_{\rm ex,2}(\check{V}) &= r_2(\check{V}) \ 
\one_{[\frac{7 c_m}{8},\infty)} \left( \frac{M+m}{2m} \check{V}\right),
\end{split}
\end{equation*}
where we have used the dummy variable $\check{V}$ for both
processes and where we suppress dependence on time in the notation $J_{m,k}$.

For trajectories in $G_m,$ the Poisson processes $J_{m,1}$ and $J_{m,2}$ provide
the bound
\begin{equation*}
\begin{split}
W_{\rm ex}(\tau^*_m) &\leq J_{m,1}  + J_{m,2}. 
\end{split}
\end{equation*} 
Therefore,
\begin{equation*}
\begin{split}
\PP(\{W_{\rm ex}(\tau^*_m) \geq \eps/6\} \cap G_m)
\leq& \, \PP(\{J_{m,1} \geq \eps/12\} \cap G_m) + \PP(\{J_{m,2} \geq \eps/12\} \cap G_m).
\end{split}
\end{equation*}
To show that 
$\PP(\{W_{\rm ex}(\tau^*_m) \geq \eps/6\} \cap G_m) \rightarrow 0$
in~\eqref{prob_decomp}, we show below that
$\PP(\{J_{m,k} \geq \eps/12\} \cap G_m) \rightarrow 0$ for $k=1,2.$

\subsubsection{Bounding $J_{m,1}$}
Applying Markov's inequality, we have 
\begin{equation*}
\begin{split}
\PP(\{J_{m,1} \geq \eps/12\} \cap G_m) &\leq \frac{12}{\eps} \EE_m(J_{m,1}).
\end{split}
\end{equation*}
We make the substitution 
$x = \frac{M+m}{2m^{1/2}} \check{V} - m^{1/2} \frac{c_m}{8}$ and
bound the integrand on the finite interval to obtain the bound
\begin{equation*}
\begin{split}
 \EE_m(J_{m,1}) &= ((t_0+z)-t_0)
\int_{0}^{\infty}
          \check{V} r_{\rm ex, 1}(\check{V}) \, d\check{V}          \\
&= C_m S_{d-1} z m^{-5/2}  \int_{\frac{2m}{M+m}\frac{7
c_m}{8}}^{\frac{2m}{M+m}\frac{9c_m}{8}}  \check{V}^2
f^1_m\left(\frac{M+m}{2m} \check{V}  - \frac{c_m}{8} \right) \, d \check{V} \\
&= \frac{8 C_m S_{d-1} z m^{1/2}}{(M+m)^3}
    \int_{\frac{3}{4} m^{1/2} c_m}^{m^{1/2} c_m}  
\left(\frac{x}{m^{1/2}} + \frac{c_m}{8}\right)^2 f^1\left(x\right) \, d x \\
&\leq \frac{8 C_m S_{d-1} z m^{1/2}}{(M+m)^3}
 \left(\frac{9 c_m}{8} \right)^2 
 \int_{0}^{\infty}  
f^1\left(x\right) \, d x \\
&\leq \frac{81 \lambda R^{d-1} S_{d-1} z}{32 (M+m)} m^{1/2} c_m^2,
\end{split}
\end{equation*}
where we recall the scaling of $f_m$~\eqref{fv1}.
We recall that $c_m = m^{-1/5}$ so that for any finite $z$ we have
$$\lim_{m \rightarrow 0} \EE_m(J_{m,1}) = 0.$$
Thus, we conclude
\begin{equation*}
\lim_{m \rightarrow 0} \PP(\{J_{m,1} \geq \eps/12\} \cap G_m) = 0.
\end{equation*}

\subsubsection{Bounding  $J_{m,2}$}
\label{sec:overlapping}
We now wish to show
\begin{equation*}
\lim_{m \rightarrow 0} \PP(\{J_{m,2} \geq \eps/12\} \cap G_m) = 0.
\end{equation*}
For this estimate, we first show the bound $\EE_m(J_{m,2}) \leq
\eps/24$ before invoking Chebyshev's inequality.  We have that
\begin{equation}
\label{jm2_1}
\begin{split}
\EE_m(J_{m,2}) &= z 
\int_{0}^\infty \check{V} \check{r}_{\rm ex, 2}(\check{V}) \, d\check{V} \\
&= C_m  S_{d-1} z m^{-5/2}  
\int_{\frac{2 m}{M+m} \frac{7 c_m}{8}}^\infty  \check{V}^2 \left[  f^1_m\left(\frac{M+m}{2m}\check{V} - \frac{c_m}{8} \right) 
 - f^1_m\left(\frac{M+m}{2m}\check{V} - \frac{c_m}{8} + \eps\right) \right] d \check{V} \\
&= \frac{8 C_m  S_{d-1} }{(M+m)^3} z m^{1/2}  
\left[ \int_{\frac{3}{4} m^{1/2} c_m}^\infty  
\left(\frac{x}{m^{1/2}} + \frac{c_m}{8}\right)^2 
f\left(x\right) \, dx \right. \\
&\qquad \qquad \left. 
- \int_{m^{1/2} (\frac{3}{4} c_m + \eps)}^\infty  
\left(\frac{x}{m^{1/2}} + \frac{c_m}{8} - \eps\right)^2 
f\left(x\right) \, dx \right] \\
&= \frac{8 C_m  S_{d-1} }{(M+m)^3} z m^{1/2}  
\left[ \int_{0}^\infty  
\left( 2 \eps \left(\frac{x}{m^{1/2}} + \frac{c_m}{8}\right) - \eps^2 \right)
f\left(x\right) \, dx \right. \\
&\qquad \qquad - \int_0^{\frac{3}{4} m^{1/2} c_m}  
\left(\frac{x}{m^{1/2}} + \frac{c_m}{8}\right)^2 
f\left(x\right) \, dx 
\left. + \int_{0}^{m^{1/2} (\frac{3}{4} c_m + \eps)}
\left(\frac{x}{m^{1/2}} + \frac{c_m}{8} - \eps \right)^2 
f\left(x\right) \, dx \right]  \\
&= \frac{4 \lambda R^{d-1}   S_{d-1}}{M} z \eps \Phi_1
+ O(c_m^2 m^{1/2}) 
\end{split}
\end{equation}
where we made two changes of variables, 
$\frac{x}{m^{1/2}}=\frac{M+m}{2m}\check{V} - \frac{c_m}{8}$ 
and $\frac{x}{m^{1/2}}=\frac{M+m}{2m}\check{V} - \frac{c_m}{8} + \eps.$ 
We recall $\lim_{m \rightarrow 0} c_m^2 m^{1/2} = 0,$ so that only
the leading order term remains in the limit.

Now, by our choice of 
$z = \min\left(M (192 \lambda R^{d-1} S_{d-1} \Phi_1)^{-1},
\frac{1}{2 (1+\|A\|)}\right)$ in
Lemma~\ref{thm:strong_extend}, we have
\begin{equation}
\label{jm_est}
\EE_m(J_{m,2}) \leq \eps/24
\end{equation} whenever $m$ is sufficiently small. 
We later need a lower bound on the expected value, $\EE_m(J_{m,2}) \geq C
\eps,$ where it is not necessary to precisely find the constant 
so long as $C$ does not depend on $\eps$ or $m.$  Such a bound follows
from the above equations and the definition of $z.$ 

Applying Chebyshev's
inequality and using~\eqref{jm_est}, we have
\begin{equation*}
\begin{split}
\PP(\{J_{m,2} \geq \eps/12\} \cap G_m) 
   &\leq \PP(\{J_{m,2} \geq \eps/12\})  \\
   &\leq \PP(\{J_{m,2} \geq 2 \EE_m(J_{m,2})\}) \\
   &= \PP(\{J_{m,2} - \EE_m(J_{m,2}) \geq \EE_m(J_{m,2})\}) \\
   &\leq \frac{ \EE_m( (J_{m,2} - \EE_m(J_{m,2}))^2)}{\EE_m(J_{m,2})^2}\\
   &= \frac{\EE_m(N_{\rm ex,2}) \EE_m(\check{V}^2)}{\EE_m(J_{m,2})^2},
\end{split}
\end{equation*}
where the final inequality follows from~\eqref{comp_var}.

We can then estimate $\EE_m(N_{\rm ex,2}) \EE_m(\check{V}^2)$ with a
similar computation as for $\EE_m(J_{m,2}):$
\begin{equation}
\label{jm2_2}
\begin{split}
\EE_m(N_{\rm ex,2}) &\EE_m(\check{V}^2) 
= \int_{t_0}^{t_0 + z} 
\int_{0}^\infty \check{V}^2 \check{r}_{\rm ex, 2}(\check{V}) \, d\check{V} dt \\
&= \frac{16 C_m  S_{d-1}}{(M+m)^4} z m^{3/2}  
\left[ \int_{\frac{3}{4} m^{1/2} c_m}^\infty  
\left(\frac{x}{m^{1/2}} + \frac{c_m}{8}\right)^3 
f\left(x\right) \, dx \right. \\
&\qquad \qquad \qquad \left. 
- \int_{m^{1/2} (\frac{3}{4} c_m + \eps)}^\infty  
\left(\frac{x}{m^{1/2}} + \frac{c_m}{8} - \eps\right)^3 
f\left(x\right) \, dx \right] \\
&\leq  \frac{12  S_{d-1}}{M^2} z m^{3/2}  
\int_{0}^\infty  
\eps \left(\frac{x}{m^{1/2}} + \frac{c_m}{8}\right)^2 
f\left(x\right) \, dx + o(m^{1/2}) \\
&\leq C z m^{1/2} + o(m^{1/2}).
\end{split}
\end{equation}
We see that this converges to zero as $m \rightarrow 0.$

Combining the above, for $m$ sufficiently small, 
we have the following estimate on $J_{m,2},$
\begin{equation*}
\begin{split}
\PP(\{ J_{m,2} \geq \eps/12 \} \cap G_m) 
&\leq \frac{\EE_m(N_{\rm ex,2}) \EE_m(\check{V}^2)}{\EE_m(J_{m,2})^2},
\end{split}
\end{equation*}
and so in view of~\eqref{jm2_1} and~\eqref{jm2_2},  
\begin{equation*}
\lim_{m \rightarrow 0} \PP(\{J_{m,2} \geq \eps/12\} \cap G_m) = 0.
\end{equation*}
Combining the estimates of this section, we conclude
\begin{equation*}
\begin{split}
\lim_{m \rightarrow 0} &\PP(\{W_{\rm ex}(\tau^*_m) \geq \eps/6\} \cap
G_m) \leq
\lim_{m \rightarrow 0} \PP(\{J_{m,1} \geq \eps/12\} \cap G_m) 
+ \lim_{m \rightarrow 0} \PP(\{J_{m,2} \geq \eps/12\} \cap G_m) = 0,
\end{split}
\end{equation*}
which completes the proof of Lemma~\ref{thm:wex}.
\end{proof}

\subsection{Slow collisions} 
\label{sec:slow}

In this section, we bound the effect of jumps of the mechanical
process caused by slow collisions.  To do so, we track the total change 
in the
relative velocity of a bath atom, $\vcen = \vb - A\qb.$  A bath
atom that undergoes a slow collision may collide with the large
particle many times; however, if the 
normal component of the bath atom's 
velocity ever exceeds $c_m$ after a collision, 
then no further collisions are possible, which follows
from arguments used in Appendix~\ref{norecollide}.  

Suppose a particular 
bath atom undergoes $k$ collisions with the large particle.
We denote by $\vcen_i$ (respectively $\vcen'_i$) for $i=1,\dots,k$ 
the relative velocity of the bath atom
before (respectively after) the $i$th collision.  We 
show below that the total change in velocity for a bath atom is bounded by
$|\vcen'_k - \vcen_1| \leq 8 c_m.$  By the preservation of linear 
momentum for an elastic collision, the change in velocity for the 
large particle due to all the collisions with this particular bath
atom is $\Vb' - \Vb = \frac{m}{M}
(\vcen'_k -\vcen_1).$  We also show below that the number of unique
bath atoms experiencing a slow collision with the large particle
in a given finite time interval is bounded
above by $C m^{-1/2} c_m$. Thus, the total expected 
change in the large particle's velocity is
bounded by $C m^{1/2} c_m^2$ which converges to zero as $m \rightarrow
0.$

\begin{lemma}
For $\tau^*_m$ be given by~\eqref{vel_close}, $G_m$ given by~\eqref{gm},
and $W_{\rm slow}$ defined above in~\eqref{wslow}, we have that
\begin{equation*} 
\lim_{m \rightarrow 0}
\PP(\{W_{\rm slow}(\tau^*_m) \geq \eps/6\} \cap G_m) = 0.
\end{equation*}
\end{lemma}

\begin{proof}
We focus now on a single bath atom and consider the set of times 
$I_{\rm rec} = \{ t_i \}_{i=1,\dots,k}$ of all collisions of this
single atom in the interval $[t_0, \tau^*_m]$ (the fact that
there are almost surely only finitely many collisions is shown in
Appendix~\ref{wellposed}).  We later estimate the total number of
distinct bath atoms that undergo slow collisions with the large
particle in the time interval.  We denote by $\eb_{{\rm n}, i}$ the
normal direction for the collision at time $t_i.$  We denote by
$\vcen_{{\rm n}, i},$ $\vcen_{{\rm n},i}'$ the normal velocity for the
particle before and after the collision at time $t_i$, and similarly
for the tangential velocity.  Finally, we let
$\vce_{{\rm n}, i} = \vcen_{{\rm n}, i} \cdot \eb_{{\rm n}, i}.$

We recall that in a slow collision, the normal velocity of the
bath atom satisfies $\vce_n \leq c_m$, whereas the tangential velocity cannot
be bounded, so we do not have a bound on $|\vb|.$   
In the following, we consider three cases, based on the size of the
tangential velocity of the bath
atom.

For the first case, we consider that $|\vcen_{{\rm t},i}| < c_m$ for
all $i=1,\dots,k.$  The sum of the change in $\vcen$ telescopes since
$\vcen_{{\rm n}, i}' - \vcen_{{\rm n}, i} = \vcen'_i - \vcen_i=
\vcen_{i+1} - \vcen_i,$ giving
\begin{equation*}
\begin{split}
\left| \sum_{i=1}^k \vcen_{{\rm n}, i}' - \vcen_{{\rm n}, i} \right| 
  &= \left| \sum_{i=1}^k \vcen_{i}' - \vcen_{i} \right| 
   = | \vcen'_k - \vcen_1| \leq 2 \sqrt{2} c_m.
\end{split}
\end{equation*}

As a second case, suppose that $|\vcen_{\rm t,1}| \geq c_m.$ 
Even with a large tangential velocity, it is possible that there is
a recollision with the large particle due to the curvature of the 
particle (see Figure~\ref{fig:slow}).  
We split the large particle into halves, with
the equator perpendicular to $\eb_{\rm t_1}$ as in
Figure~\ref{fig:slow}.  The bath atom passes the equator and due
to its high tangential velocity, which is unchanged by the collision, 
all subsequent collisions must occur
on the same hemisphere (the shaded hemisphere in
Figure~\ref{fig:slow}(a)).  We recall that the speed of the large particle
is bounded by $c_m/8$~\eqref{stop_time_mech}.  Between subsequent collisions, the bath atom
moves further in the $\eb_{{\rm t},1}$-direction, so that 
$\langle \eb_{\rm t,1}, \eb_{{\rm n},{i+1}} \rangle \leq \langle
\eb_{\rm t,1}, \eb_{{\rm n}, i} \rangle$ and therefore
$\langle \eb_{\rm t,1}, \eb_{{\rm n}, i} \rangle \leq 0$ for all $i.$  
This in turn implies monotonicity of the normal velocity for the bath
atom's collisions,
\begin{equation}
\label{monotone}
\vce_{{\rm n},{i+1}} \leq
\vce_{{\rm n}, i}'.
\end{equation}
\begin{figure}
\centerline{\input{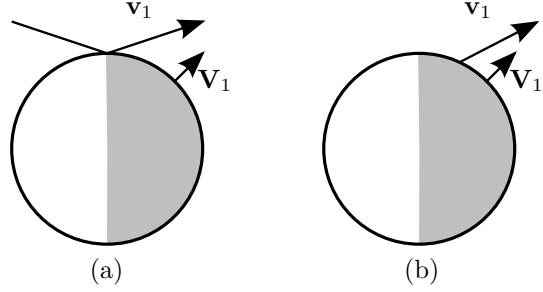} 
\hspace{1cm} %LaTeX with PSTricks extensions
%%Creator: 0.47
%%Please note this file requires PSTricks extensions
\psset{xunit=.5pt,yunit=.5pt,runit=.5pt}
\begin{pspicture}(164.31575012,178.02960205)
{
\newrgbcolor{curcolor}{0 0 0}
\pscustom[linewidth=1.77175212,linecolor=curcolor]
{
\newpath
\moveto(103.03046,138.48478205)
\lineto(162.23974,169.07980205)
}
}
{
\newrgbcolor{curcolor}{0 0 0}
\pscustom[linestyle=none,fillstyle=solid,fillcolor=curcolor]
{
\newpath
\moveto(141.80373843,167.36893765)
\lineto(164.31575668,170.18766437)
\lineto(148.99287839,153.45608913)
\curveto(149.89286282,159.12481913)(146.97186601,164.73500467)(141.80373843,167.36893765)
\closepath
}
}
\rput(118,178){$\vb_1$}
\rput(155,125){$\Vb_1$}
{
\newrgbcolor{curcolor}{0 0 0}
\pscustom[linewidth=1.77175212,linecolor=curcolor]
{
\newpath
\moveto(121.60356,121.60381205)
\lineto(145.68359,145.68389205)
}
}
{
\newrgbcolor{curcolor}{0 0 0}
\pscustom[linestyle=none,fillstyle=solid,fillcolor=curcolor]
{
\newpath
\moveto(126.7315255,137.84957853)
\lineto(147.32528572,147.36972696)
\lineto(137.80517801,126.77594736)
\curveto(136.9417065,132.45035383)(132.45536903,136.90896126)(126.7315255,137.84957853)
\closepath
}
}
{
\newrgbcolor{curcolor}{0.74901962 0.74901962 0.74901962}
\pscustom[linestyle=none,fillstyle=solid,fillcolor=curcolor]
{
\newpath
\moveto(72.56126637,145.56281416)
\curveto(112.33014407,146.11401612)(145.01385454,114.16431802)(145.56237693,74.20116612)
\curveto(146.11089931,34.23801423)(114.31651933,1.39464118)(74.54764163,0.84343922)
\curveto(74.14896156,0.83791346)(73.7502427,0.83571498)(73.35152629,0.836844)
\lineto(73.554454,73.20312669)
\closepath
}
}
{
\newrgbcolor{curcolor}{0 0 0}
\pscustom[linewidth=2.4079912,linecolor=curcolor]
{
\newpath
\moveto(145.683472,73.44373694)
\curveto(145.683472,33.54683242)(113.34064052,1.20400094)(73.443736,1.20400094)
\curveto(33.54683148,1.20400094)(1.204,33.54683242)(1.204,73.44373694)
\curveto(1.204,113.34064147)(33.54683148,145.68347294)(73.443736,145.68347294)
\curveto(113.34064052,145.68347294)(145.683472,113.34064147)(145.683472,73.44373694)
\closepath
}
}
\end{pspicture}}
 \hspace{-.5cm} (a) \hspace{3.5cm} (b)
\caption{\label{fig:slow}
A slow collision with high tangential velocity in the plane defined by
the bath atom's tangential and normal velocities.  (a) If the
tangential velocity of the first collision is large all subsequent
collisions must occur in one hemisphere, in the positive
$\eb_{t,1}$-direction.  (b) The later collisions have normal velocity
smaller and smaller, since the large component in the $\eb_{t,1}$
direction is outward pointing from the surface of the particle at the
site of subsequent collisions.}
\end{figure}

From Lemma~\ref{thm:norecollide}, any bath atom involved in a fast collision can never undergo slow collisions, so in particular we have that $|\vcen_{{\rm n},i}| \leq c_m$ for $i=1,\dots,k.$
We use the monotonicity of the normal
velocities~\eqref{monotone} to form a telescopic sum on
the normal component of the change in velocity, which becomes 
\begin{equation}
\label{eq:slow1}
\begin{split}
\left| \sum_{i=1}^k  \vcen_{{\rm n}, i}' - \vcen_{{\rm n}, i} \right| 
&\leq  \sum_{i=1}^{k-1} \left| 
   \vcen_{{\rm n},{i+1}} - \vcen_{{\rm n}, i} \right| 
+ |\vcen_{{\rm n},k}' - \vcen_{{\rm n},k} |\\
&\leq   \vce_{{\rm n},1} - \vce_{\rm n,k}  
+ |\vcen_{{\rm n}, k}' - \vcen_{{\rm n}, k} |\\
&\leq 2 c_m + 2(c_m + c_m/16) \leq 5 c_m.
\end{split}
\end{equation}
To prove the third inequality, we use the bound on 
$|\vce_{{\rm n},i}| \leq c_m$ for
$i = 0,\dots,k,$ along with the bound on the large particle
in~\eqref{stop_time_mech}, and apply the collision rule~\eqref{collision2}.

For the third and final case, consider the case that after a certain
number of collisions, $|\vcen_{{\rm t}, \ell}| \geq c_m$ whereas until
that time, we have that $|\vcen_{{\rm t}, i}| < c_m$ and $|\vcen_{{\rm
n},i}| < c_m,$ for $i=1,\dots,\ell-1.$  This combines the previous two
cases.  The change in velocity satisfies
\begin{equation*}
\begin{split}
\left| \sum_{i=1}^k  \vcen_{{\rm n}, i}' - \vcen_{{\rm n}, i} \right| 
 &= \left| \sum_{i=1}^{\ell-1} \vcen_{{\rm n}, i}' - \vcen_{{\rm n}, i} \right| + \left| \sum_{i=\ell}^{k} \vcen_{{\rm n}, i}' - \vcen_{{\rm n}, i} \right|\\
&= | \vcen_{\ell-1}' - \vcen_1 | 
 + \left| \sum_{i=\ell}^{k} \vcen_{{\rm n}, i}' - \vcen_{{\rm n}, i} \right|\\
&\leq  2 \sqrt{2} c_m + 5 c_m \leq 8 c_m.
\end{split}
\end{equation*}  
We estimate
$|\vcen_1| \leq \sqrt{2} c_m$ and $|\vcen_{\ell-1}'| \leq \sqrt{2}
c_m$ as both normal and tangential components of the relative velocity
are bounded.  We bound the final term, involving collisions after the
fast tangential velocity using the previous argument~\eqref{eq:slow1}.

Taking the maximum over the above three estimates, we conclude that
the total momentum change for the large particle due to collisions
with this single, slow particle satisfies
\begin{equation}
\label{momentum_bound}
\left| \sum_{i=1}^k M \left(\Vb_{{\rm n}, i}' 
                          - \Vb_{{\rm n}, i}\right) \right| 
= \left| \sum_{i=1}^k m \left(\vcen_{{\rm n}, i}' 
                            - \vcen_{{\rm n}, i}\right) \right|
\leq 8 m c_m.
\end{equation}

We now define $N_{\rm slow}$ as the number of first-time collisions
that are slow collisions, and we employ a double-index notation
for the collisions, so that the $j$th atom's $i$th collision has
relative normal velocity $\vcen_{{\rm n},{i,j}}$ where $1\leq j \leq
N_{\rm slow}$ and $1 \leq i \leq k_j.$  While the slow collisions are
not a Poisson jump process, we note~\eqref{bg_measure1} acts as an upper
bound for the first-time slow collisions.   (We recall that in the
mechanical process certain slow collisions are made impossible by the past
history, though no first-time slow collisions are ignored in~\eqref{bg_measure1}.)  Then we may estimate the number
of distinct bath atoms involved in slow collisions by
\begin{equation}
\label{nslow}
\begin{split}
\EE_m(N_{\rm slow}) 
&\leq z \int_{\SSph^{d-1}} \int_0^{c_m} \frac{\lambda}{m^{1/2}} R^{d-1}    (v_{\rm n} -V_{\rm n})_+  f^1_m((\vb-A\qb)_{\rm n}) \, d v_{\rm n} \, d \Omega \\
&\leq \lambda R^{d-1} S_{d-1} m^{-1/2} z \int_0^{c_m} 2 c_m f^1_m((\vb-A\qb)_{\rm n}) \, d v_{\rm n} \\
&\leq 2 \lambda R^{d-1} S_{d-1} m^{-1/2} z c_m. \\
\end{split}
\end{equation}
Combining~\eqref{momentum_bound} and~\eqref{nslow}, we finally have the estimate
\begin{equation*}
\begin{split}
\PP(\{W_{\rm slow}(\tau_m^*) \geq \eps/6\} \cap G_m) 
&\leq \PP\left(  \ \frac{m}{M} \left| \sum_{j=1}^{N_{\rm slow}} \sum_{i=1}^{k_i} 
     \vcen_{{\rm  n},{i,j}}' - \vcen_{{\rm  n},{i,j}} \right|  \geq \eps/6
                   \right)\\
&\leq \PP\left(  \frac{8 m c_m}{M} N_{\rm slow} \geq \eps/6
                \right) \\
&\leq \frac{48 m c_m}{\eps} \EE_m(N_{\rm slow}) \\
&\leq \frac{96 \lambda R^{d-1} S_{d-1} z c_m^2 m^{1/2}}{\eps} 
\rightarrow 0  \text{ as } m \rightarrow 0. \qedhere
\end{split}
\end{equation*}
\end{proof}

\subsection{Conclusion of the proof}

Combining the above estimates on $W_{\rm slow}$ and $W_{\rm ex},$ 
we conclude that the left hand side of~\eqref{prob_decomp}
converges to zero in the limit as $m \rightarrow 0.$  This shows that
Lemma~\ref{thm:strong_extend} holds, and we then conclude
that Lemma~\ref{thm:strong} follows.  Over finite time intervals, we can then
conclude that $(\Qtb_m, \Vtb_m)_\tint - (\Qb_m, \Vb_m)_\tint \pconv 0$ 
as $m \rightarrow 0.$ Combined with the convergence result Lemma~\ref{thm:ma_to_ou}, we can therefore conclude the proof of our main result,
Theorem~\ref{thm:main}.

\section{The mechanical process is well-posed}
\label{wellposed}

In this section, we show that the mechanical process is almost surely
well-posed for all times $t < \tau_m$, where $\tau_m$ is the stopping
time defined in~\eqref{stop_time_mech}.  Particularly, we show that as
long as the large particle stays bounded, the degenerate cases of having a
multicollision or having infinitely many collisions in finite time
both occur with zero probability.  We further show that almost surely
$\tau_m > 0.$ We note that this well-posedness result is slightly
weaker than the corresponding result in~\cite{durr80,durr83} for a
heat bath with zero background flow since that result does not require
a stopping time and is valid over finite time intervals that are
uniform with respect to the initial condition.  In contrast, the
stopping time $\tau_m$ depends on the initial condition.

We first show that whenever the particle is restricted to a ball
$|\Qb_m| < \ell$ for $\ell \in \NN,$ then there are almost surely only
finitely many atoms that can possibly interact with the particle in a
given time interval $[0,T]$, since only finitely many enter the ball.
It is therefore sufficient to study the well-posedness of the 
dynamics for finite systems.  After restricting to a finite system, we show that the mechanical
process is almost surely well defined.  Since $|\Qb_m| <
c_m/8$ for $t \in [0,\tau_m),$ we conclude that the mechanical process
is well-posed at least until the stopping time $\tau_m.$

\subsection{Finitely many atoms within a finite radius}

In the following, we work with fixed mass $m$ for the bath atoms,
which we absorb into other constants.

\begin{lemma} 
\label{thm:finiteness}
For any finite time interval $[0,T]$ and radius $\ell \in \NN,$ 
there are almost surely only finitely many bath atoms that enter the ball
$|\qb| < \ell$ in the time interval  $[0,T].$
\end{lemma}

\begin{proof}
For a given initial configuration, let $N_k$ be the number of bath 
atoms such that $k < |\qb| \leq k+1$ and $|\vcen|
\geq \eps k,$ where we choose $\eps = \frac{1}{2} e^{- ( ||A|| +
\frac{1}{2}) T}.$     
We then have
\begin{equation*}
\begin{split}
\EE_m(N_k) &= \int_{k \leq |\qb| \leq k+1} \int_{|\vcen| \geq \eps k} 
         \lambda_m f_m(\qb,\vcen + A\qb) \, d \vcen \, d \qb  \\
&= C [ (k+1)^d - k^d] S_{d-1} \int_{|\vcen| \geq \eps k} 
	\frac{|\vcen|^{d+1}}{|\vcen|^{d+1}} f(m^{1/2} \vcen) \, d \vcen  \\
&\leq \frac{C k^{d-1}}{(\eps k)^{d+1}} \int_{|\vcen| \geq m^{1/2} \eps k} 
	|\vcen|^{d+1} f(\vcen) \, d \vcen  \\
&\leq \frac{C}{k^2},
\end{split} 
\end{equation*} 
where we note that $C$ depends on $\eps,$ but not on $k.$
Then, almost surely $\sum_{k=1}^\infty N_k < \infty.$
Suppose an atom has initial coordinates $(\qb_0, \vcen_0)$ where $k \leq
|\qb_0| \leq k+1$ and $|\vcen_0| < \eps k$ for some $k \in \NN.$ 
 We have from the
dynamics~\eqref{nonHam2} that 
\begin{equation*}
\begin{split}
\frac{d}{dt} \left( \frac{|\qb|^2}{2} \right) &=\qb^T A \qb + \qb^T
\vcen_0 \\
&\geq - ||A|| \ |\qb|^2 - \frac{1}{2} |\qb|^2 - \frac{1}{2} |\vcen_0|^2.
\end{split}
\end{equation*} 
This implies 
\begin{equation*}
\begin{split}
|\qb(t)|^2 &\geq \left(|\qb_0|^2 + \frac{|\vcen_0|^2}{2 ||A|| + 1} \right) 
e^{- (2 ||A|| + 1) t} - \frac{ |\vcen_0|^2 }{ 2 ||A||+1}.\\
&\geq |\qb_0|^2 e^{- (2 ||A|| + 1) t} - |\vcen_0|^2 \\
&\geq k^2 e^{- (2 ||A|| + 1) t} - \eps^2 k^2 \\
&= \frac{3}{4} k^2 e^{- (2 ||A|| + 1) T}.
\end{split} 
\end{equation*}  
Let $\ell \in \NN$ be fixed.  
We choose $k_0$ sufficiently large such that $\frac{\sqrt{3}}{2}
k_0 e^{- (||A|| + \frac{1}{2}) T}  > \ell.$   Then, the only atoms that can
enter the ball $|\qb| < \ell$ must either be one of the finitely many such
that $|\qb_0| < k_0$ or be such that there exists $k \geq k_0$ such that 
$k < |\qb_0| \leq k+1$ and $|\vcen_0|
\geq \eps k.$
\end{proof}

In a finite mechanical system, the first collision occurs at a time
$t_0>0,$ 
since at time $t=0,$ no bath atoms are touching the large particle.
Upon choosing $ \ell = \frac{c_m}{8},$ this implies that the stopping time
$\tau_m$~\eqref{stop_time_mech} for the mechanical process satisfies
$\tau_m > 0.$  (Recall that that after Definition~\eqref{stop_time_mech},
we have assumed that $m$ is sufficiently small so that the initial 
conditions lie within the box that defines the stopping criteria.

\subsection{A finite system is well-posed}

In the following, we consider a system composed of the large particle
along with $N \in \NN$ bath atoms, in the finite time interval
$[0,\tau_m).$   We show that there are almost surely only finitely
many collisions, none of which are multicollisions.  For any initial
condition of the mechanical process we define the process on the time
interval $[0, \tau^*)$ where $\tau^*$ is the first occurrence of either a
multicollision or an accumulation of infinitely many collisions, or the
first time the large particle leaves the box used in the definition of
$\tau_m$ (see~\eqref{stop_time_mech}).  By definition, the mechanical
process is well-posed on the time interval $[0, \tau^*),$ and 
we have $\tau^* \leq
\tau_m.$  From Lemma~\ref{thm:finiteness}, we can conclude that 
$\tau^* > 0,$ and we show below that almost surely $\tau^* = \tau_m.$

We note that because of the bath atom dynamics~\eqref{nonHam1}, the
kinetic energy of the mechanical process is not preserved.  When
restricting to a finite system, we can control the growth of the
energy.
\begin{lemma}
\label{kin_en}
In the time interval $[0, \tau^*),$ the kinetic energy of the 
mechanical system at time $t$,
\begin{equation*}
K(t) = \frac{1}{2} M \Vb(t)^T \Vb(t) 
     + \frac{1}{2} \sum_{j=1}^N m \vb_j(t)^T \vb_j(t),
\end{equation*}
satisfies 
$K(t) \leq K(0) e^{2 \|A\| t}.$ 
\end{lemma}
\begin{proof}
During an elastic collision, the kinetic energy is unchanged.  Between
collisions, we have
\begin{equation*}
\begin{split}
\frac{d}{dt} K(t) &= M \Vb(t)^T \frac{d}{dt} \Vb(t) 
     + \sum_{j=1}^N m \vb_j(t)^T \frac{d}{dt} \vb_j(t) \\
&= \sum_{j=1}^N m \vb_j(t)^T A \vb_j(t),
\end{split} 
\end{equation*}
and we thus have
$|\frac{d}{dt} K(t)| \leq 2 \|A\|
K(t).$  Integrating gives the desired result.  
\end{proof}
As a corollary, we have that for every initial condition, the speed
$\vmax = \left(\frac{K(0) e^{2 \|A\| \tau^*}}{m} \right)^{1/2}$ 
bounds the maximum speed of any bath atom:
\begin{equation} 
\label{vmax}
\max_{t \in [0,\tau^*)} \max_{j=1,\dots,N} |\vb_i(t)| \leq \vmax.
\end{equation} 

We let $B_N(\ell, K)$ denote the set of initial conditions where $\tau^*
< \tau_m$, that is, those that lead to either a multicollision or
infinitely many collisions, where $N$ is the number of bath atoms in
the system, where the large particle is restricted to $|\Qb|
\leq \ell$, and where the total kinetic energy satisfies $\sup_{[0,T]}
K(t) \leq K.$  If we show that the set $B_N(\ell, K)$ has zero measure,
then the set of initial conditions leading to such 'bad' collisions, which
is given by
\begin{equation*} 
\bigcup_{N=1}^\infty \bigcup_{\ell = 1}^\infty \bigcup_{K = 1}^\infty 
B_N(\ell, K),
\end{equation*}   
has zero measure.  

We first restrict ourselves to the case of a system with the large
particle and a single bath atom.

\begin{lemma}
\label{lem:tan}
For all initial conditions in $B_1(\ell,K),$ there are infinitely many
collisions in $[0, \tau^*)$ and the normal component of the atom and particle
velocities before the $i$th collision
satisfies
\begin{equation*}
\lim_{i \rightarrow \infty} |v_{{\rm n}, i} - V_{{\rm n}, i}| = 0.
\end{equation*}
\end{lemma}
\begin{proof}
Since there is only one bath atom, the only possible bad collision is
the accumulation of infinitely many collisions. 
Suppose a collision occurs with relative velocity $v_{{\rm n}, i} -
V_{{\rm n}, i} = \eps,$ for some $\eps >0.$   Immediately after the
collision $v_{{\rm n}, i}' - V_{{\rm n}, i}' = -\eps.$
Since the system is composed of only a single particle and a single
bath atom, in order for another collision to occur, the bath atom must
accelerate until $(\vb(t) - \Vb(t))\cdot \eb_{{\rm n}, i} \geq 0.$  

From the bound on the bath atom velocity in~\eqref{vmax}, 
\begin{equation*}
\frac{d}{dt} \vb(t) = A \vb \leq \|A\| \vmax.
\end{equation*}
We conclude that there exists $\delta >0$ such that the time between
collisions satisfies $t_{i+1} - t_i \geq \delta \eps.$  Thus, an
accumulation of collisions implies that 
$\lim_{i \rightarrow \infty} |v_{{\rm n}, i} - V_{{\rm n}, i}| = 0.$
\end{proof}

\begin{lemma}
\label{thm:inf_collide}
The set $B_1(\ell,K)$ has measure zero.
\end{lemma}
\begin{proof}
For any $\eps >0$, let $\tau_\eps^*$ denote the time of the first collision where
$|v_{{\rm n}, i} - V_{{\rm n}, i}| < \eps.$   
For $P \in \NN,$ let $\Gamma^P_{\eps,p} \subset B_1(\ell,K)$ be the set of initial
conditions where $\tau_\eps^* \in [\frac{p}{P} \tau_m, \frac{p+1}{P} \tau_m),$ for
$p = 0, \dots, P-1.$  By Lemma~\ref{lem:tan}, an accumulation of collisions must 
lead to slow
collisions, so that $B_1(\ell, K) \subset \cup_{p=0}^{P-1} \Gamma^P_{\eps, p}.$

Up to time $\tau_\eps^*,$ the dynamics is well-defined, and 
the flow map $\phi_t$ for the mechanical process preserves Lebesgue measure.  
For $p=1,\dots,P-1,$ $\phi_{\frac{p}{P}\tau_m}(\Gamma^P_{\eps, p}) \subset \Gamma^P_{\eps, 0},$
so that $|\Gamma^P_{\eps, p}| \leq |\Gamma^P_{\eps, 0}|.$ 

For initial conditions in $\Gamma^P_{\eps, 0},$ we may suppose that $P^{-1}$ is 
sufficiently small so that $\tau_\eps^*$ is the time of the first 
collision (recall from Lemma~\ref{lem:tan} that a collision with $|v_{{\rm n}, i} - V_{{\rm n}, i}| \geq \eps$ leads to a minimum time between collisions 
$t_{i+1} - t_i \geq \delta \eps$). 
Mapping back from a collision at time $\tau^*_\eps,$ we see that the 
bath atom's initial position satisfies
$|\qb - \Qb| \leq \frac{\vmax \tau_m}{P}$ and the initial velocity satisfies  
$|v_{{\rm n}, i} - V_{{\rm n}, i}| \leq \eps e^{\|A\| \frac{\tau_m}{P}} \vmax.$
Because of the bounds on $\Qb$ and $K,$ we then conclude that
$|\Gamma^P_{\eps, 0}| \leq \frac{C_1 \eps}{P} e^{\frac{C_2}{P}},$
where $C_1, C_2\in \RR$ are independent of $\eps$ and $P.$
Summing over $p=0,\dots,P-1,$ we have
\begin{equation*}
| B_1(\ell,K)| \leq \lim_{\eps \rightarrow 0} \lim_{P \rightarrow \infty}
\left| \bigcup_{p=0}^{P-1} \Gamma^P_{\eps, p} \right| = 0. \qedhere
\end{equation*}
\end{proof}

\begin{lemma}
\label{thm:multi_collide}
The subset $M \subset B_N(\ell,K)$ of initial conditions where the first 
bad collision is a multicollision has measure zero.  
\end{lemma}

\begin{proof}
Intuitively, the set of all multicollisions at a given instant $\tau^*$ has phase space co-dimension greater 
than or equal to two, so that the set of initial conditions leading to a multicollision
has zero measure.  We note that a multicollision could involve one or more bath atoms
that have an accumulation of infinitely many collisions at time $\tau^*.$

For $P\in \NN,$ we let $\Gamma^P_{\eps, p}$ be the set of initial conditions such that
$\tau^* \in [\frac{p}{P} \tau_m, \frac{p+1}{P} \tau_m),$ where $p=0,\dots,P-1,$
where we recall that $\tau^*$ is the time of the first multicollision.  
We note that since the flow map of the mechanical process preserves Lebesgue measure,
$|\Gamma^P_{\eps, p}| \leq |\Gamma^P_{\eps, 0}|,$ so we only need to bound the measure of
$\Gamma^P_{\eps, 0}.$  

For initial conditions in $\Gamma^P_{\eps, 0},$ 
two or more bath atoms must satisfy
$| \qb - \Qb| \leq \frac{\vmax \tau_m}{P}$, 
so that we have $|\Gamma^P_{\eps, 0}| \leq \frac{C}{P^2}.$
We thus have
\begin{equation*}
| M | = \lim_{P \rightarrow \infty}
\left| \bigcup_{p=0}^{P-1} \Gamma^P_{\eps, p} \right| = 0. \qedhere
\end{equation*}
\end{proof}

\begin{lemma}
The set of initial conditions leading to a bad collision has
measure zero.
\end{lemma}

\begin{proof}
We first consider the full set $B_N(\ell,K).$  If the first bad
collision in such a system is not a multicollision, then there
is a time $\delta$ such that in the interval $[\tau^* - \delta, \tau^*),$
the only collisions that the large particle experiences are with the
bath atom that experiences the bad collision at $\tau^*.$  In this case,
the arguments from Lemma~\ref{thm:inf_collide} show that the set of 
such initial conditions has measure zero.  Therefore, using Lemma~\ref{thm:multi_collide},
we conclude that $B_N(\ell,K)$ has measure zero as well.

We then conclude that 
\begin{equation*}\left|\bigcup_{N=1}^\infty \bigcup_{\ell = 1}^\infty \bigcup_{K = 1}^\infty 
B_N(\ell, K)
  \right| = 0,
\end{equation*}
and that almost surely $\tau^* = \tau_m$ and so the mechanical 
process is well-posed until at least the
stopping time $\tau_m.$
\end{proof}

\section{Convergence of the generator for the laminar heat bath}
\label{lam_gen}

In this section, we compute the limiting SDE for the two laminar heat
bath models defined in Section~\ref{sec:lam1}.  The computations
parallel those of Appendix~\ref{ma_to_ou}, and we omit some of the
details found there.  The first heat bath model has velocity
distribution of the form~\eqref{fvq_shear}, and the second has
distribution of the form~\eqref{fvq_multi_shear} where we consider the
specific shear flow
\begin{equation*}
A = \left[ \begin{array}{ccc}
0&s&0\\
0&0&0\\
0&0&0\\
\end{array} \right],
\end{equation*}
for some $s \in \RR.$
We do not carry out a full convergence proof as we do for the model
proposed in Section~\ref{meas_dyn} --- although such a proof could likely
be carried out based on the same arguments as above --- since the limiting
dynamics obtained by the following formal computations has undesirable
traits for sampling nonequilibrium states.  

\subsection{Unidirectional flow}
\label{uni_flow}
We start with a background bath whose atoms are distributed initially
according to~\eqref{fvq_shear}, where all bath atom velocities at initial
time are oriented in the $\eb_1$-direction.  In fact, rather than working 
with~\eqref{fvq_shear} and its Gaussian distribution, we work with
\begin{equation}
\label{lam_law}
d \mu_m(\qb,v_1) = \lambda_m f_m( v_1 - s q_2) \, d \qb \, d v_1,
\end{equation}  
where we have the scaling $f_m(\vce) = m^{1/2} f(m^{1/2} \vce)$,
where the probability density function $f(\vce)$ is assumed to be
decreasing in $\vce$ with four finite moments~\eqref{phi_i_2}, and
where we recall $q_2 = Q_2 - R \eb_n \cdot \eb_2$.  Note
that we have eliminated $v_2$ and $v_3$ in the above expression since
they are identically zero.  We recover~\eqref{fvq_shear} upon choosing
$f(x) = Z^{-1} \exp\left(-\frac{\beta}{2} x^2 \right).$

The normal velocity for a bath atom before a collision
is $v_n = v_1 \eb_1 \cdot \eb_{\rm n},$ and we keep the distinction
between slow and fast collisions of Definition~\ref{fast}.  The rate of fast collisions
experienced by the particle is 
\begin{equation*}
\begin{split} 
r_m(v_1, \eb_{\rm n}, \Qb, \Vb) \, dv_1 \, d \Omega \, dt 
=& \lambda_m R^2 \ 
 H( (v_1 \eb_1 - s q_2)\cdot \eb_{\rm n} - c_m) \\
& \max(v_{\rm n} - V_{\rm n}, 0) f_m( v_1 - s q_2) \, dv_1 \, d\Omega
\, dt.
\end{split} 
\end{equation*}
In contrast to~\eqref{bg_measure1}, the rate function here includes a
Heaviside function to ignore slow collisions. 
We let $\Vh = V_n' - V_n$, which satisfies, in view
of~\eqref{collision2}
\begin{equation}
\label{collision_lam}
\begin{split}
\Vh &= \frac{2m}{M+m} (v_1 \ \eb_1 \cdot \eb_{\rm n} - V_{\rm n}),  \\
\Vhm &= \frac{2m}{M+m} (c_m + s q_2 \  \eb_1 \cdot \eb_{\rm n}- V_{\rm n}),  \\
\end{split}
\end{equation}
where $\Vhm$ denotes the minimum velocity jump size (due to the cutoff for slow
collisions).  For $\Vh \geq \Vhm,$ the measure on jumps 
$\Vh \eb_{\rm n}$ is
\begin{equation*}
\begin{split} 
\hat{r}_m(\Vh, \eb_{\rm n}, \Qb, \Vb)  &= 
	\lambda_m R^2 \left(\frac{M+m}{2m}\right)^2  
	\frac{\Vh}{|\eb_1 \cdot \eb_{\rm n}|} \  
    f_m\left(\frac{1}{\eb_1 \cdot \eb_{\rm n}} 
	\left[\frac{M+m}{2m} \Vh + V_{\rm n}\right] 
  -  s q_2 \right) .
\end{split} 
\end{equation*}
Using this jump measure, we define a Markov process as in
Definition~\ref{def:mark}.

We write the generator for the Markov process 
\begin{equation*}
L_m \psi = \Vb \cdot \nabla_Q \psi + \int_{\SSph^{2}} \int_\Vhm^\infty 
              [\psi(\Qb, \Vb + \Vh \eb_{\rm n}) -\psi(\Qb, \Vb)] \ 
		  \hat{r}_m(\Vh, \eb_{\rm n}, \Qb, \Vb) \, d \Vh \, d \Omega.
\end{equation*}
From the scaling of bath atom velocities and associated
law~\eqref{lam_law}, $\EE_m(|\vb|)$ is
$O(m^{-1/2})$ and by~\eqref{collision_lam} we have that   
$\EE_m(\Vh)$ is $O(m^{1/2})$. We expand the generator for the Markov
process in powers of
$\Vh$ and find to leading order  
\begin{equation}
\label{Lm_expand_lam}
\begin{split}
L_m \psi =& \Vb \cdot \nabla_Q \psi(\Qb, \Vb) 
+ C_m \left(  I_1 \cdot \nabla_V \psi(\Qb,\Vb) + \frac{1}{2} I_2 : \nabla_V^2 \psi(\Qb,\Vb) \right) + O(m^{1/2}),
\end{split}
\end{equation}
where 
\begin{align} 
\label{cm2}
C_m &= \lambda R^2
\left(\frac{M+m}{2}\right)^2,  \\
\notag
I_1 &= 
m^{-5/2} \int_{\SSph^{2}} \int_\Vhm^\infty 
	\eb_{\rm n} \frac{\Vh^2}{|\eb_1 \cdot \eb_{\rm n}|} 
f_m\left( \frac{1}{\eb_1 \cdot \eb_{\rm n}} \left[\frac{M+m}{2m} \Vh 
	 + V_{\rm n}\right] - s q_2\right) 
	d\Vh \, d\Omega,\\
\notag
I_2 &= 
m^{-5/2} \int_{\SSph^{2}} \int_\Vhm^\infty 
	(\eb_{\rm n} \otimes \eb_{\rm n}) \frac{\Vh^3}{|\eb_1 \cdot \eb_{\rm n}|}  
f_m\left( \frac{1}{\eb_1 \cdot \eb_{\rm n}} \left[\frac{M+m}{2m} \Vh 
	  + V_{\rm n}\right] - s q_2\right) 
	d\Vh \, d\Omega,
\end{align}
where $q_2 = Q_2 - R (\eb_{\rm n} \cdot \eb_2).$  
Note that estimating the remainder term follows as in
Appendix~\ref{sec:mark_rem}.

In both integrals $I_1$ and $I_2$, we substitute 
\begin{equation*}
\begin{split} 
x &= \frac{1}{|\eb_1 \cdot \eb_{\rm n}|} 
   \left[ \frac{M+m}{2 m^{1/2}} \Vh + m^{1/2} V_{\rm n} \right] 
- m^{1/2} s \frac{\eb_1 \cdot \eb_{\rm n}}{|\eb_1 \cdot \eb_{\rm n}|}
q_2,\\
\xm &= \frac{m^{1/2} c_m}{|\eb_1 \cdot \eb_n|},
\end{split} 
\end{equation*}
where $\xm > 0$ is the minimum value of $x.$
The drift term, $I_1,$ becomes
\begin{equation}
\label{i1_lam}
\begin{split}
I_1 =&  m^{-1/2} \int_{\SSph^{2}} \int_{\xm}^\infty 
	\eb_{\rm n} {\left(|\eb_1 \cdot \eb_{\rm n}| x - m^{1/2} V_{\rm n} 
	      + m^{1/2} s (\eb_1 \cdot \eb_{\rm n}) q_2 \right)^2}  
  \frac{8}{(M+m)^3}  f\left( \frac{|\eb_1 \cdot \eb_{\rm n}|}
             {\eb_1 \cdot \eb_{\rm n}} x\right) \, dx \, d\Omega\\
=&\frac{8 m^{-1/2}}{M^3}  \left[ \int_{\SSph^{2}} \int_{\xm}^\infty 
	\eb_{\rm n} {\left(|\eb_1 \cdot \eb_{\rm n}| x \right)^2} 
	f\left( \frac{|\eb_1 \cdot \eb_{\rm n}|}{\eb_1 \cdot \eb_{\rm n}} 
        x\right) \,	dx \, d\Omega \right. \\
&+ \left. 2 \int_{\SSph^{2}} \int_{0}^\infty \eb_{\rm n} 
     \left(- m^{1/2} V_{\rm n} + m^{1/2} s (\eb_1 \cdot \eb_{\rm n}) q_2  \right)   \left(|\eb_1 \cdot \eb_{\rm n}| x \right)  
	f\left( \frac{|\eb_1 \cdot \eb_{\rm n}|}{\eb_1 \cdot \eb_{\rm n}} 
				 x\right) \, dx \, d\Omega \right] + O(m^{3/10}) \\
=&\frac{16 \Phi_1}{M^3}  \left[ \int_{\SSph^{2}} \eb_{\rm n} 
     \left(- \Vb \cdot \eb_{\rm n} + s (\eb_1 \cdot \eb_{\rm n}) q_2  \right)  
     |\eb_1 \cdot \eb_{\rm n}|   
	 \, d\Omega \right] + O(m^{3/10}),
\end{split} 
\end{equation}
where we have used the fact that $O(m^{-1/2})$-term integrates to zero 
and have extended the integration region with the estimate 
\mbox{$\int_0^{\xm} |\eb_1 \cdot \eb_{\rm n}| x f(x) \, dx
\leq  m^{1/2} c_m \int_\RR f(x) \, dx = O(m^{3/10}).$}  Note that $q_2
= Q_2 - R \eb_n \cdot \eb_2$
depends on $\eb_n,$ but the $\eb_n$-dependence annihilates as we show below
in~\eqref{sphere_lam_3}.

Substituting $x$ into $I_2,$ we compute the diffusion coefficient to be
\begin{equation}
\label{i2_lam}
\begin{split} 
I_2 &= \frac{16}{M^4} \int_{\SSph^{2}} \int_{\xm}^\infty 
     (\eb_{\rm n} \otimes \eb_{\rm n}) \left(|\eb_1 \cdot \eb_{\rm n}| x\right)^3 
	f\left( \frac{\eb_1 \cdot \eb_{\rm n}}{|\eb_1 \cdot
              \eb_{\rm n}|} x\right) 
	dx \, d\Omega + O(m^{1/2})\\
&= \frac{16 \Phi_3}{M^4} 
\int_{\SSph^{2}} (\eb_{\rm n} \otimes \eb_{\rm n}) |\eb_1 \cdot \eb_{\rm n}|^3  
	\, d\Omega + O(m^{1/2}),\\
\end{split} 
\end{equation}
where we have used the inequality $\int_0^{\xm} (|\eb_1 \cdot
\eb_{\rm n}| x)^3 f(x) \, dx
\leq  m^{3/2} c_m^3 \int_\RR f(x) \, dx = O(m^{1/2}).$

Using the following expressions for the spherical integrals
\begin{align}
\label{sphere_lam_1}
\int_0^{2\pi} \int_0^\pi  
	(\eb_{\rm n} \otimes \eb_{\rm n}) |\eb_1 \cdot \eb_{\rm n}|  \ \sin \phi  
	\, d\phi \, d\theta 
&= \left(
  \begin{array}{ccc}
    \pi& 0 & 0 \\
    0& \frac{\pi}{2}  & 0 \\
    0& 0&\frac{\pi}{2} \\
  \end{array}
\right), \\
\label{sphere_lam_2}
\int_0^{2\pi} \int_0^\pi  
	(\eb_{\rm n} \otimes \eb_{\rm n}) |\eb_1 \cdot \eb_{\rm n}|^3  \ \sin \phi  
	\, d\phi \, d\theta
&= \left(
  \begin{array}{ccc}
    \frac{2\pi}{3}& 0 & 0 \\
    0& \frac{\pi}{6}  & 0 \\
    0& 0&\frac{\pi}{6} \\
  \end{array}
\right), \\
\label{sphere_lam_3}
\int_0^{2\pi} \int_0^\pi  
	(\eb_{\rm n} \otimes \eb_{\rm n}) |\eb_1 \cdot \eb_{\rm n}|  
    (\eb_2 \cdot \eb_{\rm n}) \ \sin \phi \ d\phi \, d\theta &= 0,
\end{align}
we deduce from~\eqref{i1_lam} and~\eqref{i2_lam} that 
\begin{equation*}
\begin{split}
I_1 &= - \frac{16 \pi \Phi_1}{M^3}  
\left(
  \begin{array}{ccc}
    1& 0 & 0 \\
    0& \frac{1}{2}  & 0 \\
    0& 0&\frac{1}{2} \\
  \end{array}
\right) (\Vb - s (\eb_1 \otimes \eb_2) \Qb), \qquad 
I_2 = \frac{16 \pi \Phi_3}{M^4} 
\left(
  \begin{array}{ccc}
    \frac{2}{3}& 0 & 0 \\
    0& \frac{1}{6}  & 0 \\
    0& 0&\frac{1}{6} \\
  \end{array}
\right).
\end{split}
\end{equation*}

Combining these terms with~\eqref{Lm_expand_lam} and~\eqref{cm2},
writing the matrices in terms 
of tensor products, choosing $f$ from~\eqref{gauss1},
and taking the limit as $m \rightarrow 0,$ 
the generator converges in the sense of~\eqref{gen_converge}
to the generator
\begin{equation*}
\begin{split}
L \psi &= \Vb \cdot \nabla_Q \psi + 
\frac{2 \sqrt{2 \pi} \lambda R^2}{\sqrt{\beta} M} 
\left[ - \left(\eb_1 \otimes \eb_1 + \frac{1}{2}  
                   (I - \eb_1 \otimes \eb_1)\right)
(\Vb - s (\eb_1 \otimes \eb_2) \Qb) \cdot \nabla_V \psi \right. \\ 
&\qquad + \left. \frac{1}{\beta M} 
\left(\frac{2}{3} \eb_1 \otimes \eb_1 + \frac{1}{6} (I - \eb_1 \otimes
\eb_1)\right) : \nabla_V^2 \psi) \right].\\
\end{split}
\end{equation*}
This is the generator for the anisotropic SDE
\begin{equation}
\label{ou_shear_app}
\begin{split}
d \Qb &= \Vb dt \\
M d \Vb &= - \gamma (\Vb - s (\eb_1 \otimes \eb_2) \Qb) dt + \sigma d\Wb 
       = - \gamma (\Vb - A \Qb) dt + \sigma d\Wb \\
\end{split}
\end{equation}
where
\begin{equation*}
\gamma =  \frac{2 \sqrt{2\pi} \lambda R^2}{\sqrt{\beta}}
    \left( \begin{array}{ccc} 
            1 &0&0\\0&\frac{1}{2}&0\\0&0&\frac{1}{2}
           \end{array}\right) 
\qquad \text{ and } \qquad 
\sigma = \left[
       \frac{4 \sqrt{2\pi} \lambda R^2}{\sqrt{\beta^3}}
    \left( \begin{array}{ccc} 
            \frac{2}{3} &0&0\\0&\frac{1}{6}&0\\0&0&\frac{1}{6}
           \end{array}\right)          \right]^{1/2}.
\end{equation*} 
This is the dynamics~\eqref{ou_shear}  discussed in
Section~\ref{sec:lam1}.
As mentioned there, we find the inconsistency of~\eqref{ou_shear_app}
when $s=0$ with the Langevin
dynamics for zero background flow~\eqref{lang1} unacceptable.  
As an attempt to remove the
anisotropy observed in~\eqref{ou_shear_app},  
we add laminar flows in {\em{each}} coordinate
direction in the following section.

\subsection{Superimposed baths}
\label{sec:superimposed}

We now consider having three different laminar flows, one for each coordinate 
direction $\eb_i,$  for $i=1,2,3.$   We define the probability
distributions 
\begin{equation*}
f_i(v) = 
 Z_i^{-1} \exp\left(-\frac{1}{2 \theta_i^2}  v^2\right),      
\qquad \text{ where } Z_i={\sqrt{2 \pi} \theta_i}, \qquad 
i = 1,2,3,
\end{equation*}
where we allow for different temperatures for
each of the bath measures. 
The initial condition of the bath atoms is chosen according to the
measure
\begin{equation*}
\begin{split}
d \mu_m(\qb,\vb) &= 
 \frac{\lambda_m m^{1/2}}{3} \left( f_{1}(m^{1/2} (v_1 - s q_2)) \delta(v_{2})
\delta(v_{3})  \right.\\
&\left. +f_{2}(m^{1/2} v_2) \delta(v_{1}) \delta(v_{3})
 +f_{3}(m^{1/2} v_3) \delta(v_{1}) \delta(v_{2}) \right) 
d \vb  d \qb, \qquad \qb, \vb \in \RR^3,
\end{split}
\end{equation*}
where $\delta(x)$ is the Dirac distribution. 

Upon splitting the integral linearly, the calculations of
Appendix~\ref{uni_flow} give us the limiting generator
\begin{equation*}
\begin{split}
L \psi &= \Vb \cdot \nabla_Q \psi 
   - M^{-1} (\gamma \Vb - \tilde u)\cdot \nabla_V \psi 
   + \frac{1}{2} M^{-2} \sigma \sigma^T :\Delta_V \psi 
\end{split}
\end{equation*}
where 
\begin{equation*}
\begin{split}
\gamma &=  \frac{2 \sqrt{2\pi} \lambda R^2}{3}
	    \left( \begin{array}{ccc} 
	  \theta_1+ \frac{1}{2} \theta_2+
	\frac{1}{2}\theta_3&0&0\\
	0& \frac{1}{2} \theta_1+  \theta_2+
	\frac{1}{2}\theta_3&0\\
	0&0&\frac{1}{2} \theta_1+ \frac{1}{2} \theta_2+
	\theta_3 
           \end{array}\right), \\
\tilde u &=  \frac{2 \sqrt{2\pi} \lambda R^2}{3}
    \left( \begin{array}{c} 
      {s Q_2}{\theta_1}\\
			{0}{}\\
			{0}{}\\
           \end{array}\right), \\
\sigma \sigma^T &=   
	\frac{4 \sqrt{2\pi} \lambda R^2}{3}
	    \left( \begin{array}{ccc} 
	  \frac{2}{3}\theta^3_1 + \frac{1}{6} \theta^3_2+
	\frac{1}{6}\theta^3_3&0&0\\
	0& \frac{1}{6}\theta^3_1 + \frac{2}{3}\theta^3_2 +
	\frac{1}{6}\theta^3_3&0\\
	0&0&\frac{1}{6} \theta^3_1+ \frac{1}{6} \theta^3_2+
	\frac{2}{3}\theta^3_3
           \end{array}\right). 
\end{split}
\end{equation*} 
The corresponding SDE is
\begin{equation*}
\begin{split}
d \Qb &= \Vb dt, \\
M d \Vb &= - \gamma \Vb dt + \tilde u dt + \sigma d\Wb.
\end{split}
\end{equation*} 
As an example, let us take $\theta_i =
\frac{1}{\sqrt{\beta}}$ for $i=1,2,3.$  Then 
we have for the limiting equation
\begin{equation*}
\begin{split}
d \Qb &= \Vb dt,\\
M d\Vb &= - \gamma \left(\Vb - \frac{1}{2} s Q_2 \eb_1\right) dt +
\left(\frac{\gamma}{\beta}\right)^{1/2} d\Wb,
\end{split}
\end{equation*} 
where 
$\gamma = \frac{4 \lambda R^2 \sqrt{2 \pi}}{3 \sqrt{\beta}}.$
This is~\eqref{ou_shear2}, which is discussed in
Section~\ref{sec:lam1}.  This system satisfies the
fluctuation-dissipation relation only for inverse temperature
$\tilde{\beta}= 2 \beta.$   That is, the temperature of the large
particle is half that of the bath.  More problematically, there is a
factor of $1/2$ on $s,$  which means that the large particle only
feels half the average velocity of the  heat bath.  This is caused by
the superposition of multiple laminar flows.  We note that there is no
choice of $\theta_1, \theta_2,$ and $\theta_3$ so that the resulting
SDE is both isotropic and has response $\tilde{u}$ equal to the input
background motion.

\section*{Acknowledgements}

This work is supported in part by the NSF Mathematical Sciences
Postdoctoral Research Fellowship  and by the Agence Nationale de la
Recherche, under grants ANR-06-CIS6-006 (PARMAT) and
ANR-09-BLAN-0216-01 (MEGAS).  
 The authors would like to thank Claude Le Bris for many fruitful
discussions.

\bibliographystyle{abbrv}

\end{document}